\documentclass{article}

\usepackage{arxiv}

\usepackage[utf8]{inputenc} % allow utf-8 input
\usepackage[T1]{fontenc}    % use 8-bit T1 fonts
\usepackage{hyperref}       % hyperlinks
\usepackage{url}            % simple URL typesetting
\usepackage{booktabs}       % professional-quality tables
\usepackage{amsfonts}       % blackboard math symbols
\usepackage{nicefrac}       % compact symbols for 1/2, etc.
\usepackage{microtype}      % microtypography
\usepackage{lipsum}         % Can be removed after putting your text content
\usepackage{graphicx}
\usepackage[round]{natbib}

\usepackage{amsmath}
\usepackage{amssymb}
\usepackage{amsthm}
\usepackage{enumerate}
\usepackage[inline]{enumitem}

\graphicspath{{./}{./figures/}}

\newcommand{\vol}{\mathrm{vol}}
\newcommand{\E}{\mathbb E}
\newcommand{\var}{\mathrm{var}}

\newcommand{\Normal}{\mathcal N}
\newcommand{\Beta}{\mathrm{Beta}}

\newcommand{\bbR}{\mathbb R}
\newcommand{\bbE}{\mathbb E}

\newcommand{\bbN}{\mathbb N}

\newcommand{\calY}{\mathcal Y}

\newcommand{\calM}{\mathcal M}
\newcommand{\calF}{\mathcal F}

\newcommand{\OneFOne}{{}_1F_1}
\newcommand{\naturalset}{\Omega}
\newcommand{\Gammadist}{\mathrm{Gamma}}
\newcommand{\Binomial}{\mathrm{Binomial}}
\newcommand{\Poisson}{\mathrm{Poisson}}
\newcommand{\CR}{C_\alpha} % Confidence region of the mean
\newcommand{\CRnp}{C_\alpha^{\mathrm{np}}} % confidence region of the natural parameter in natural exponential families
\newcommand{\AR}{A_\alpha} % Acceptance region of the mean
\newcommand{\ARnp}{A_\alpha^{\mathrm{np}}} % Acceptance region of the natural parameter
\newcommand{\loA}{\underline{A}_\alpha}
\newcommand{\upA}{\overline{A}_\alpha}
\newcommand{\FAB}{\mathrm{FAB}}
\newcommand{\MAP}{\mathrm{MAP}}
\newcommand{\statfunc}{\calF}
\newcommand{\dHaus}{d_{\mathrm{H}}}% Hausdorff distance between closed sets
\newcommand{\bX}{\mathbf{X}}
\newcommand{\bx}{\mathbf{x}}
\newcommand{\bbeta}{\boldsymbol{\beta}}
\newcommand{\bY}{\mathbf{Y}}

\DeclareMathOperator*{\argmax}{arg\,max}
\DeclareMathOperator*{\argmin}{arg\,min}

\newtheorem{theorem}{Theorem}[section]
\newtheorem{lemma}{Lemma}[section]
\newtheorem{corollary}{Corollary}[section]
\newtheorem{proposition}{Proposition}[section]
\newtheorem{definition}{Definition}[section]
\newtheorem{remark}{Remark}[section]
\newtheorem{example}{Example}[section]
\newtheorem{assumption}{Assumption}[section]

\newlist{enuminline}{enumerate*}{1}
\setlist[enuminline]{label=(\roman*)}

\usepackage[capitalize,nameinlink,noabbrev]{cleveref}       % smart cross-referencing

\Crefname{assumption}{Assumption}{Assumptions}

\title{Bayes-assisted Confidence Regions:\\Focal Point Estimator and Bounded-influence Priors}

\date{}

\author{
  Stefano Cortinovis\\
	Department of Statistics\\
	University of Oxford\\
	\texttt{cortinovis@stats.ox.ac.uk}\\
	\And
  Fran\c cois Caron\\
	Department of Statistics\\
	University of Oxford\\
	\texttt{caron@stats.ox.ac.uk}
}

\renewcommand{\headeright}{}
\renewcommand{\undertitle}{}
\renewcommand{\shorttitle}{Bayes-assisted Confidence Regions}

\hypersetup{
pdftitle={Bayes-assisted Confidence Regions: Focal Point Estimator and Bounded-influence Priors},
pdfsubject={stat.ME},
pdfauthor={Stefano Cortinovis; Fran\c cois Caron},
pdfkeywords={Frequentist Coverage, Robust Bayes, Shrinkage, Heavy-tailed priors, Horseshoe},
}

\begin{document}
\maketitle

\begin{abstract}
	The Frequentist, Assisted by Bayes (FAB) framework constructs confidence regions that leverage prior information about parameter values. FAB confidence regions (FAB-CRs) have smaller volume for values of the parameter that are likely under the prior while maintaining exact frequentist coverage. This work introduces several methodological and theoretical contributions to the FAB framework. For Gaussian likelihoods, we show that the posterior mean of the mean parameter is contained in the FAB-CR. More generally, this result extends to the posterior mean of the natural parameter for likelihoods in the natural exponential family. These results provide a natural Bayes-assisted estimator to be reported alongside the FAB-CR. Furthermore, for Gaussian likelihoods, we show that power-law tail conditions on the marginal likelihood induce robust FAB-CRs that are uniformly bounded and revert to standard frequentist confidence intervals for extreme observations. We translate this result into practice by proposing a class of shrinkage priors for the FAB framework that satisfy this condition without sacrificing analytic tractability. The resulting FAB estimators equal prominent Bayesian shrinkage estimators, including the horseshoe estimator, thereby establishing insightful connections between robust FAB-CRs and Bayesian shrinkage methods.
\end{abstract}


%% Introduction
\section{Introduction}
\label{sec:introduction}

Modern large-scale studies, such as RNA-seq screens or web-scale A/B tests, often involve estimating thousands of effects, most of which are expected to be near zero, with only a sparse minority plausibly large.
In such settings, using a classical confidence region (CR) for the effect of interest may be inefficient.
Indeed, the expected volume of such classical CR is the same for small effects and implausibly large effects.
This inefficiency is especially costly when thousands of CRs are reported and interpreted side-by-side.

Pratt's \emph{Frequentist, Assisted by Bayes} (FAB) paradigm \citep{Pratt1961,Pratt1963,Yu2018} addresses this mismatch by finding the CR procedure $C_\alpha(\cdot)$ that minimizes the expected volume
\begin{align}
  R(\CR) &= \int_{\Theta} \E[\vol(\CR(Y))|\theta] \pi(d\theta)
  \label{eq:Bayesexpectedvol}
\end{align}
under some prior $\pi$, subject to the usual frequentist coverage constraint.
Intuitively, the working prior $\pi$ acts as an importance weight, allowing the optimal FAB confidence regions (FAB-CRs) to be narrower where parameters are deemed plausible and wider elsewhere.
This has inspired a range of methodological developments and applications \citep{Brown1995, Puza2006, Farchione2008, Hoff2019, Kabaila2022, Woody2022,Hoff2023}.
An illustration is given in \cref{fig:fab_gaussian_gaussian} for a Gaussian likelihood and prior.
\begin{figure}
  \centering
  \includegraphics[width=\textwidth]{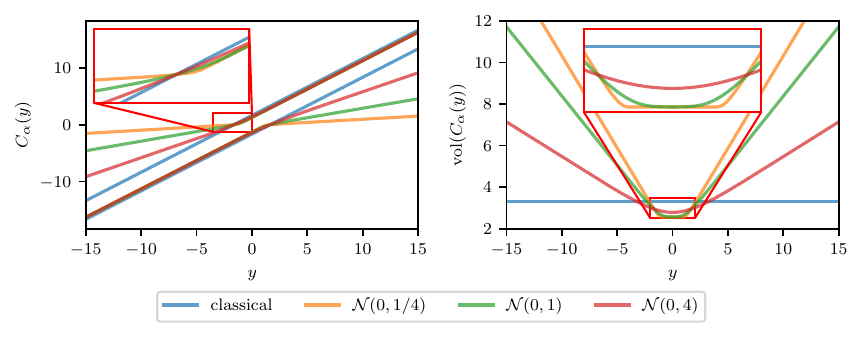}
  \caption{
    Comparison of the $z$-interval and the FAB-CR procedure under the Gaussian prior $\Normal(0, \tau^2)$ for $\tau^2 \in \{1/4,1,4\}$, when $Y\mid \theta\sim \Normal(\theta, 1)$ for $\alpha = 0.1$.
    The left and right panels show the extrema of each CR and their widths as functions of the observed data $y$.
  }
  \label{fig:fab_gaussian_gaussian}
\end{figure}

Existing FAB constructions (see \cref{sec:fabbackground}) still leave two important issues open.
First, the literature does not provide a canonical estimator to report alongside a FAB confidence region; we call this estimator the \textit{focal point}.
Second, while FAB-CRs always satisfy the coverage requirement, their volume may be arbitrarily large in the presence of disagreement between the prior $\pi$ and the data.
For instance, this is the case for the FAB-CR under a Gaussian prior applied to a Gaussian likelihood with known variance, as shown in \cref{fig:fab_gaussian_gaussian} (right).
Although fixes have been proposed in the literature \citep[e.g.,][]{Farchione2008,Hoff2019}, they either rely on employing specific improper priors or use ad hoc modifications of the standard FAB-CR procedure, and a general framework for obtaining uniformly bounded confidence regions is currently missing. We provide a unified treatment that addresses these shortcomings.
\begin{enumerate}[label=(\roman*),itemsep=2pt,topsep=2pt]
  \item \textbf{Uniformly bounded FAB regions.}
    For Gaussian likelihoods, we show that priors with tails heavier than Gaussian guarantee uniformly bounded FAB-CR volume; polynomial tails further make the regions revert to classical $z$-intervals for extreme observations.
  \item \textbf{Focal points and shrinkage estimators.}
    We formalize the \emph{focal point} of nested FAB-CRs and prove it equals the posterior mean under $\pi$ for Gaussian likelihoods;
    a closed-form extension covers the entire natural exponential family.
  \item \textbf{Shrinkage-prior blueprint.}
    We provide a scale-mixture class of priors that yields FAB regions which (a) promote sparsity, (b) are uniformly bounded, and (c) revert to classical intervals when prior--data conflict occurs.
\end{enumerate}

Our approach builds on the existing rich literature on Bayesian robustness \citep{DeFinetti1961,Lindley1968,Strawderman1971,Berger1980,Ohagan1981,Park2008,Caron2008,Carvalho2010,Griffin2010,Griffin2011,Armagan2011,Polson2012,OHagan2012}, and in particular the work of \citet{Dawid1973,Pericchi1992} and \citet{Pericchi1995} on priors with bounded or vanishing influence.

The remainder of this article is organized as follows.
\Cref{sec:fabbackground} provides background on the FAB approach.
In \cref{sec:gaussianlikelihood}, we consider the Gaussian likelihood case, for which we identify the focal point $\widehat\theta^\FAB(y)$, as well as sufficient conditions to obtain uniformly bounded FAB-CRs.
In \cref{sec:shrinkagepriors}, we describe a class of priors that lead to a tractable, robust FAB-CR procedure, and draw connections with the Bayesian shrinkage literature.
In \cref{sec:exponentialfamily}, we consider the more general case of likelihoods in the natural exponential family, and provide an expression for the corresponding focal point.
An application to linear regression is presented in \cref{sec:linearregression}.
Finally, \cref{sec:discussion} briefly discusses related literature and potential avenues for future work.
The Supplementary Material contains the proofs and additional derivations and figures. All sections, figures, and equations in the supplementary material are prefixed with `S' for clarity.

%% FAB background
\section{FAB decision-theoretic framework}
\label{sec:fabbackground}
Let $Y\in\mathcal Y$ have density $f_{\theta}$ w.r.t.~the Lebesgue measure, indexed by a parameter of interest $\theta\in\Theta\subseteq\mathbb R^{d}$.
For a fixed $\alpha\in(0,1)$ we seek a procedure $C_\alpha:\;y\mapsto C_\alpha(y)\subseteq\Theta$ satisfying the \emph{exact coverage} requirement
\begin{align}
  \Pr(\theta \in \CR(Y)\mid \theta=\theta_0) = 1 - \alpha
  \label{eq:coverage}
\end{align}
for any $\theta_0\in\Theta$.
In this context, a natural measure of the efficiency of a confidence region (CR) procedure $\CR$ is its expected volume
\[
  \E[\vol(\CR(Y))|\theta] = \int_\calY \vol(\CR(y)) f_{\theta}(y) dy,
\]
where $\vol(\CR(Y))=\int_{\theta_0\in \CR(Y)} d\theta_0$.
In practice, when one has knowledge about plausible values for the unknown parameter $\theta$, it may be desirable to construct a valid confidence region $\CR$ that attains smaller volumes for values of $\theta$ that are thought to be likely, at the expense of larger volumes for values of $\theta$ that are considered less likely.
Following \citet{Pratt1961,Pratt1963}, this may be achieved through a distribution $\pi$ on $\theta$, treated as a random variable, for which one aims to find the confidence region procedure $\CR$ that minimizes the (Bayes) expected volume \eqref{eq:Bayesexpectedvol} under the (frequentist) constraint \eqref{eq:coverage}.
This approach, termed ``Frequentist, Assisted by Bayes'' by \citet{Yu2018}, has been employed to derive smaller valid confidence regions for parameters of interest in the presence of prior information. Using the Ghosh--Pratt identity \citep{Ghosh1961,Pratt1961}, the Bayes expected volume \eqref{eq:Bayesexpectedvol} can be written as
\[
  R(\CR) = \int_\Theta \Pr(\theta_0 \in \CR(Y))  d\theta_0,
\]
where $\Pr(\theta_0 \in \CR(Y))$ is computed w.r.t.~the marginal density of $Y$ under the prior $\pi$, i.e., $f(y) = \int_\Theta f_\theta(y) \pi(d\theta)$.
By defining the acceptance regions $\AR(\theta_0) = \{y \mid \theta_0 \in \CR(y)\}$, the CR that minimizes $R(\CR)$ may be obtained by choosing, for all $\theta_0 \in \Theta$, the acceptance regions that maximize $\Pr(Y \notin \AR(\theta_0))$ under the coverage constraint \eqref{eq:coverage} and inverting them, i.e.,
\begin{equation}
  \CR(y) = \{\theta_0 \mid y \in \AR(\theta_0)\}.
  \label{eq:confidenceregionCgeneral}
\end{equation}
By the Neyman--Pearson lemma, the optimal acceptance regions are given by
\begin{equation}
  \AR(\theta_0)=\left\{y\mid \log\left(\frac{f(y)}{f_{\theta_0}(y)}\right)\leq k_{\alpha}(\theta_0)\right\},
  \label{eq:acceptanceregion_neymanpearson}
\end{equation}
where $k_{\alpha}(\theta_0)$ is the smallest value such that $\Pr(Y\in \AR(\theta)\mid \theta=\theta_0)=\Pr(\theta\in \CR(Y)\mid \theta=\theta_0)=1-\alpha$.
A self-contained derivation of the steps outlined above is provided in \cref{sec:derivationCgeneral}.

To illustrate the application of the FAB approach, consider the Gaussian likelihood case $Y \mid \theta \sim \Normal(\theta, \sigma^2)$, where $\sigma^2$ is known.
If one believes that $\theta$ is likely to be close to zero, they may wish to apply FAB under a Gaussian prior $\pi = \Normal(0, \tau^2)$, for some $\tau > 0$.
\Cref{fig:fab_gaussian_gaussian} compares the resulting FAB-CRs with the standard $z$-interval for the mean of a Gaussian likelihood with known variance, which is given by $y \pm \sigma \Phi^{-1}(1-\alpha/2)$, where $\Phi^{-1}$ is the quantile function of a standard normal random variable.

As desired, the FAB-CRs achieve smaller volume than the standard interval when $y$ is close to zero, and the gain in efficiency around the origin is controlled by the choice of $\tau$.
However, while FAB-CRs have correct frequentist coverage, their volume under the Gaussian prior grows without bound as the disagreement between the prior and observed data increases, i.e., $\vol(\CR(y))\to \infty$ as $|y|\to\infty$.
That is, in general, FAB may result in a CR whose volume is not uniformly bounded, which may be undesirable in practice, as noted in \cref{sec:introduction}.

\paragraph*{Focal point estimator.}
The construction of valid confidence regions for parameters of interest is deeply rooted in the statistical estimation of the same quantities.
For instance, many standard CLT-based confidence intervals follow from the asymptotic normality of the corresponding maximum likelihood estimator, which is then contained in the resulting interval at all confidence levels $(1 - \alpha)$, thereby representing its \textit{focal point}.
More formally, assume that the nested regions $\CR(y)$ concentrate around a single point as the error level $\alpha$ increases; for $d=1$, a sufficient condition is $\lim_{\alpha\to 1}\vol(\CR(y))=0$.
When the latter holds for a FAB-CR, we denote its \textit{focal point} as the intersection of the nested closed confidence regions
\begin{equation}
  \widehat\theta^\FAB(y) = \bigcap_{\alpha\in(0,1)} \CR(y),
  \label{eq:focalpointFABCR}
\end{equation}
which is well-defined; see \citep{Molchanov2017,Taraldsen1996}. As we will see in \cref{sec:exponentialfamily}, alternative definitions of $\widehat\theta^\FAB(y)$ may be considered when $\bigcap_{\alpha\in(0,1)} \CR(y)$ is a non-singleton set.
In this context, the random variable $\widehat\theta^\FAB(Y)$, which always lies within the FAB-CR $\CR(Y)$, may naturally be used as a Bayes-assisted estimator for the parameter of interest $\theta$ to be reported alongside the FAB-CR.

\paragraph*{Notations.}
For two real-valued functions $f$ and $g$ defined on $\bbR$, we write $f(x)\sim g(x)$ as $x\to\infty$ for $\lim_{x\to\infty} \frac{f(x)}{g(x)}=1$.
For a subset $C$ of $\bbR$ and $y\in\bbR$, $C+y=\{x+y \mid x \in C\}$.
For two closed subsets $C_1$ and $C_2$ of $\bbR$, let $\dHaus$ be the Hausdorff distance defined by $\dHaus(C_1,C_2)=\max\{\sup_{x\in C_1}\inf_{y\in C_2} |x-y|, \sup_{y\in C_2}\inf_{x\in C_1} |x-y| \}$.
For a collection of closed subsets $(C_1(y))_{y\in\bbR}$ and a closed subset  $C_2$ of $\bbR$, we write $\lim_{y\to\infty}C_1(y)= C_2$ for $\lim_{y\to\infty}\dHaus(C_1(y),C_2)= 0$.
When $C_2=[a,b]$ is a closed interval, for some $a<b$, $\lim_{y\to\infty}C_1(y)= C_2$ if, for all $\epsilon\in(0,\frac{b-a}{2})$, there is $y_0$ such that $[a+\epsilon,b-\epsilon]\subseteq C_1(y) \subseteq [a-\epsilon,b+\epsilon]$ for all $y>y_0$.

%% Gaussian likelihood
\section{Gaussian likelihood with known variance}
\label{sec:gaussianlikelihood}
Consider here a Gaussian likelihood,
\begin{equation}
  f_\theta(y)=\frac{1}{\sqrt{2\pi\sigma^2}}\exp\left(-\frac{(y-\theta)^2}{2\sigma^2}\right)
  \label{eq:gaussian_likelihood}
\end{equation}
where $y\in\bbR$, $\theta\in\bbR$, and $\sigma>0$ is assumed to be known.
Given that the parameter of interest $\theta$ is a scalar, the corresponding FAB-CR $C_\alpha(y) \subseteq \bbR$ is one-dimensional, and the volume $\vol(C_\alpha(y))$ is simply its length.
However, for consistency with the notation used elsewhere, we will continue to refer to $\vol(C_\alpha(y))$ as the volume of the confidence region, even in this univariate case.
\begin{assumption}
  \label{assumpt:priorpi}
  Assume that the prior distribution $\pi(d\theta)$ is a positive Radon measure on $\bbR$, with full support, and such that
  $0<\int_\bbR f_\theta(y)\pi(d\theta)<\infty$ for any $y\in\bbR$. 
\end{assumption}
\begin{remark}
  For most priors discussed in this paper, $\pi(d\theta)$ is a proper probability distribution on $\bbR$, with $\int_\bbR \pi(d\theta)=1$.
  Nonetheless, \cref{assumpt:priorpi} also allows for the use of improper priors with $\int \pi(d\theta)=\infty$, such as $\pi(d\theta)=d\theta$.
  In both cases, the finiteness of $f(y) = \int_\bbR f_\theta(y)\pi(d\theta)$ guarantees that the posterior distribution $\pi(d\theta\mid y)=f_\theta(y)\pi(d\theta) / f(y)$ is proper.
  However, note that $f(y)$ is not a proper density function if $\pi(d\theta)$ is not a probability distribution. In the sequel, even when $\pi$ is improper, we refer to the corresponding $f(y)$ as the marginal likelihood of $y$ under $\pi(d\theta)$ with a slight abuse of language.
  In the improper case, the expected volume \eqref{eq:Bayesexpectedvol} may be infinite for all confidence regions satisfying the coverage condition; one can, however, equivalently minimize the expected volume difference with a classical interval $\widetilde R(C_\alpha)=\int_{\Theta} \left\{\E[\vol(\CR(Y))|\theta]-2\sigma \Phi^{-1}(1-\alpha/2)\right\} \pi(d\theta)$ \citep{Farchione2008}.
  This leads to confidence regions of the same form \eqref{eq:confidenceregionCgeneral} and \eqref{eq:acceptanceregion_neymanpearson}.
  Finally, the degenerate case $\pi(d\theta)=\delta_{\theta_1}$ for some $\theta_1\in\bbR$ has been covered in detail by \citet{Pratt1961,Pratt1963}, and, hence, is excluded from \cref{assumpt:priorpi}.
\end{remark}
Denote the log-likelihood of $y$ and the log-marginal likelihood of $y$ under some prior $\pi$ satisfying \cref{assumpt:priorpi} by $\ell_{\theta_0}(y)=\log f_{\theta_0}(y)$ and $\ell(y)=\log f(y)$, respectively.
Let $\lambda_{\theta_0}(y) = \ell(y)-\ell_{\theta_0}(y)$ denote the log-likelihood ratio, and recall that the FAB confidence region $\CR(y)$ is obtained by inversion of the acceptance region $\AR(\theta_0)=\left \{ y\mid \lambda_{\theta_0}(y)\leq k_\alpha(\theta_0)\right \}$ through \cref{eq:confidenceregionCgeneral}.
Let
\begin{align}
  \widehat\theta(y)=\E[\theta|y]=y+\sigma^2 \ell'(y)
  \label{eq:posteriormean}
\end{align}
be the posterior mean of $\theta$ given $y$ under the prior $\pi$.
The following result states that the acceptance regions $\AR(\theta_0)$ are intervals, whose bounds are continuous and differentiable functions, and that the FAB-CR $\CR(y)$ contains the posterior mean under the prior $\pi$.
That is, for Gaussian likelihoods, the FAB focal point $\widehat\theta^\FAB(y)$ coincides with the posterior mean $\widehat\theta(y)$.
\begin{theorem}
    \label{thm:postmeaninCI}
  Assume that $\pi(d\theta)$ satisfies \cref{assumpt:priorpi}.
  Then, for any $\theta_0\in\bbR$, the log-likelihood ratio $\lambda_{\theta_0}(y) = \ell(y)-\ell_{\theta_0}(y)$ is a strictly convex function in $y$ with
  $$
    \arg\min_{y\in\bbR}\lambda_{\theta_0}(y) = \widehat\theta^{-1}(\theta_0).
  $$
  Let $\alpha\in(0,1)$. The acceptance region $\AR(\theta_0)$ is an interval $[\loA(\theta_0),\upA(\theta_0)]$ with $\loA(\theta_0)<\widehat\theta^{-1}(\theta_0)<\upA(\theta_0)$, and where $\loA(\theta_0)$ and $\upA(\theta_0)$ are the unique pair satisfying
  \begin{align}
    \lambda_{\theta_0}(\upA(\theta_0))&=\lambda_{\theta_0}(\loA(\theta_0)), \label{eq:boundarycondition_1} \\
    \Phi\left(\frac{\upA(\theta_0)-\theta_0}{\sigma}\right)-\Phi\left(\frac{\loA(\theta_0)-\theta_0}{\sigma}\right)&=1-\alpha. \nonumber
  \end{align}
  Define
  \begin{align}
    \label{eq:wfunction}
    w_\alpha(\theta_0)=\frac{1}{\alpha}\Phi\left(\frac{\loA(\theta_0)-\theta_0}{\sigma}\right)\in(0,1),
  \end{align}
  where $\Phi(z)$ is the CDF of a standard normal random variable.
  Then, for any $y\in\bbR$, the confidence region $\CR(y)$ can be expressed in terms of $w_\alpha(\theta_0)$ as
  \begin{align}
    \label{eq:confidenceregionCw}
    \CR(y) = \{\theta_0\mid \loA(\theta_0)=\theta_0-\sigma \Phi^{-1}(1-\alpha w_\alpha(\theta_0)) \leq y \leq \theta_0 + \sigma \Phi^{-1}(1-\alpha(1-w_\alpha(\theta_0))) = \upA(\theta_0) \},
  \end{align}
  where the functions $\loA(\theta_0)$, $\upA(\theta_0)$ and $w_\alpha(\theta_0)$ are all continuously differentiable on $\bbR$.
  Lastly, the focal point of the FAB-CR $\CR(y)$ is uniquely defined as in \cref{eq:focalpointFABCR} and is given by
  $$
    \widehat\theta^\FAB(y)=\widehat\theta(y)=\mathbb E[\theta\mid y].
  $$
\end{theorem}
\Cref{thm:postmeaninCI} reparameterizes the acceptance regions $\AR(\theta_0)$ of the size-$\alpha$ tests inducing the FAB-CR $\CR(y)$ in terms of the function $w_\alpha(\theta_0)$, which is defined in \cref{eq:wfunction}. This quantity, called the \textit{spending function} by \citet{Hoff2019} and the \textit{tail function} by \citet{Puza2006}, determines how the error budget $\alpha$ is split between the left and right tails of $\AR(\theta_0)$ so as to minimize the Bayes expected risk \eqref{eq:Bayesexpectedvol}.
Computationally, the FAB-CR $\CR(y)$ is obtained via \cref{eq:confidenceregionCw} by evaluating $w_\alpha(\theta_0)$ through a root-finding procedure that solves \cref{eq:boundarycondition_1} expressed in terms of $w_\alpha(\theta_0)$ as \looseness=-1
\begin{equation}
  \lambda_{\theta_0}\left(\theta_0 + \sigma \Phi^{-1}(1-\alpha(1-w_\alpha(\theta_0)))\right) = \lambda_{\theta_0}\left(\theta_0-\sigma \Phi^{-1}(1-\alpha w_\alpha(\theta_0))\right). \label{eq:boundarycondition_1_w}
\end{equation}
The main result of \cref{thm:postmeaninCI} can be nicely illustrated through the $p$-value function \citep{Fraser2019,Schweder2016} (also known as confidence curve \citep{Birnbaum1961} or significance function \citep{Fraser1991})
\begin{equation*}
  p_y(\theta_0)=\sup \{\alpha\in(0,1)\mid \theta_0\in \CR(y)\}.
\end{equation*}
As $\CR(y)=\{\theta_0 \mid p_y(\theta_0)\geq \alpha\} $, this function allows one to visualize the nested confidence regions over the whole confidence range; additionally, the focal point is a maximum of the $p$-value function, with $p_y(\widehat\theta^\FAB(y))=1$. \cref{fig:fab_gaussian_gaussian_cc} compares the $p$-value functions corresponding to the FAB-CRs obtained under the same Gaussian priors as in \cref{fig:fab_gaussian_gaussian} against the one associated with the standard $z$-interval for three observed values of $y$.
\begin{figure}
  \centering
  \includegraphics[width=.9\textwidth]{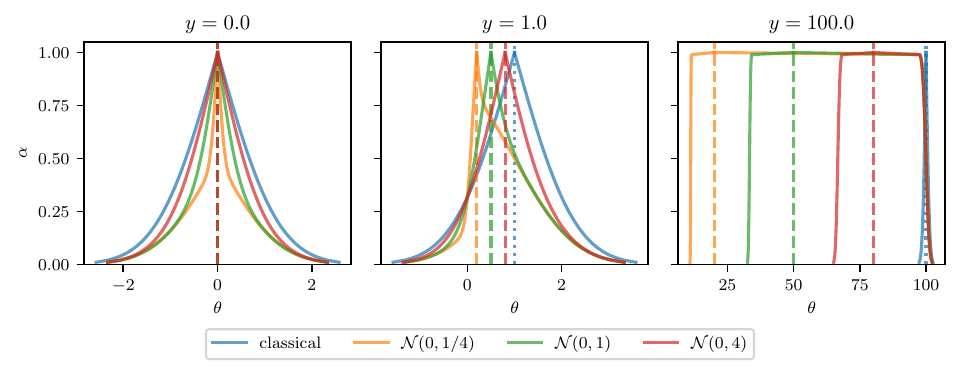}
  \caption{$p$-value functions for the classical $z$-interval and the FAB-CR procedure under the same priors and likelihood of \cref{fig:fab_gaussian_gaussian} when observing $y \in \{0, 1, 100\}$.}
  \label{fig:fab_gaussian_gaussian_cc}
\end{figure}
The dashed vertical lines indicate the posterior means under the three priors, while the dotted vertical lines indicate the MLE, $y$, which is the focal point of the $z$-interval.
The posterior means do indeed coincide with the focal points of the FAB-CRs, as stated in \cref{thm:postmeaninCI}.
This fact can be generalized to likelihoods in a natural exponential family, as shown in \cref{sec:exponentialfamily}.
Moreover, \cref{fig:fab_gaussian_gaussian_cc} shows that the Gaussian FAB-CRs lead to significantly smaller regions than the $z$-interval for values of $y$ that are likely under the prior.
However, as disagreement between the Gaussian prior and the data increases, the volume of the FAB-CRs grows unbounded.
We now consider sufficient conditions on the prior $\pi(d\theta)$ for the confidence region to be uniformly bounded.

\begin{theorem}\label{thm:robustCR}
  Assume that $\pi(d\theta)$ satisfies \cref{assumpt:priorpi}.
  Additionally, assume that the marginal likelihood $f(y)$ under $\pi(d\theta)$ has tails of the form
  \begin{align}
    f(y) &\sim \gamma|y|^{-\delta}\exp\left(-\frac{\kappa}{\sigma} |y|\right)\text{ as } y\to\pm\infty\label{eq:regvar1}
  \end{align}
  for some exponent $\delta\geq 0$, some parameter $\kappa \geq0$ and some constant $\gamma>0$, and that the derivative of $f(y)\exp\left(\kappa|y|/\sigma\right)$, defined on $\bbR\backslash\{0\}$, is ultimately monotone as $y\to \pm \infty$.
  Let $\alpha\in (0,1)$. Then,
  \begin{align*}
    \sup_{y\in\bbR} \vol(\CR(y))<\infty
  \end{align*}
  and
  \begin{align}
    \lim_{y\to\infty}  \CR(y)-y &= \left[-\sigma \Phi^{-1}(1-\alpha c_\alpha)\ ,\ \sigma \Phi^{-1}(1-\alpha (1-c_\alpha)) \right], \label{eq:lim_infCR}\\
    \lim_{y\to-\infty} \CR(y)-y &= \left[-\sigma \Phi^{-1}(1-\alpha (1-c_\alpha))\ ,\ \sigma \Phi^{-1}(1-\alpha c_\alpha) \right], \label{eq:lim_supCR}
  \end{align}
  where the convergence is with respect to the Hausdorff distance on closed subsets of $\mathbb R$, $c_\alpha=g_\alpha^{-1}(-2\kappa)\in(0,\frac{1}{2}]$, where $g_\alpha^{-1}:\bbR\to(0,1)$ is the inverse of the strictly increasing function $g_\alpha:(0,1)\to\bbR$, defined as
  \begin{align}
    g_\alpha(\omega)=\Phi^{-1}(\alpha\omega)-\Phi^{-1}(\alpha(1-\omega)) .
    \label{eq:galpha}
  \end{align}
  Moreover, the focal point $\widehat\theta^\FAB(y)$, which coincides with the posterior mean by \cref{thm:postmeaninCI}, satisfies
  \begin{equation*}
    \lim_{y\to\infty} ( y-\widehat\theta^\FAB(y))=\lim_{y\to - \infty} (\widehat\theta^\FAB(y)-y) = \kappa\sigma.
  \end{equation*}
\end{theorem}

\begin{remark}
  \cref{thm:robustCR} may be slightly generalized by replacing the constant $\gamma$ with a slowly varying function (e.g., log, power log, etc.), in which case the proof applies similarly.
\end{remark}
As described by \cref{eq:regvar1}, \cref{thm:robustCR} unveils the crucial role of the tails of the marginal likelihood $f(y)$ in ensuring that the corresponding FAB-CRs are uniformly bounded.
In particular, it indicates that the prior $\pi$ should be chosen so that $f(y)$ exhibits heavier tails than a Gaussian random variable.
This class includes exponential ($\kappa > 0$, $\delta = 0$), power-law ($\kappa=0$, $\delta > 1$), and improper ($\kappa = \delta = 0$) tails.
While all these choices successfully result in uniformly bounded FAB-CRs, a case that deserves special attention is the one with $\kappa = 0$. In this case, $c_\alpha = g_\alpha^{-1}(0) = 1 / 2$ for any $\alpha \in (0, 1)$, and we have the following corollary.
\begin{corollary}\label{cor:vanishing_influence}
  Assume that $\pi(d\theta)$ satisfies \cref{assumpt:priorpi} and that the marginal likelihood $f(y)$ under $\pi(d\theta)$ has tails of the form \eqref{eq:regvar1} with $\kappa = 0$.
  Then, as $|y|\to\infty$, the influence of the prior on the focal point $\widehat\theta^\FAB(y)$ vanishes, and the FAB confidence region $\CR(y)$ converges to the standard $z$-interval.
  Formally,
  \begin{equation*}
    \lim_{y\to\pm\infty} \widehat\theta^\FAB(y)-y = 0,
  \end{equation*}
  and, for any $\alpha\in(0,1)$,
  \begin{equation*}
    \lim_{y\to\pm\infty}  \CR(y)-y = \left[-\sigma \Phi^{-1}(1-\alpha/2)\ ,\ \sigma \Phi^{-1}(1-\alpha/2) \right],
  \end{equation*}
  where the convergence is with respect to the Hausdorff distance on closed subsets of $\mathbb R$.
\end{corollary}
The effect of the tail parameter $\kappa$ on the resulting FAB-CR mirrors that on the corresponding credible interval \citep{Dawid1973,Pericchi1995}.
For this reason, in the sequel, we refer to the cases $\kappa = 0$ and $\kappa > 0$ with the terms \textit{vanishing} and \textit{bounded} influence, respectively, as in \citet{Pericchi1995}.

As highlighted by \cref{thm:wbounded} in the proof of \cref{thm:robustCR} in the supplementary material, the limiting behavior of the weight function $w_\alpha(\theta)$ as $|\theta| \to \infty$ plays a crucial role in determining the behavior of the FAB-CRs through \cref{thm:robustCR}.
\Cref{fig:fab_gaussian_w} shows the weight functions $w_\alpha(\theta)$ for the Gaussian prior, as well as for other choices of priors with vanishing and bounded influence that will be formally introduced in \cref{sec:shrinkagepriors}.
Priors for which $w_\alpha(\theta)$ is bounded away from $0$ and $1$ lead to uniformly bounded FAB-CRs, which revert to the standard $z$-interval as $|y| \to \infty$ if $w_\alpha(\theta)$ converges to $1/2$.
\begin{figure}
  \centering
  \includegraphics[width=0.8\textwidth]{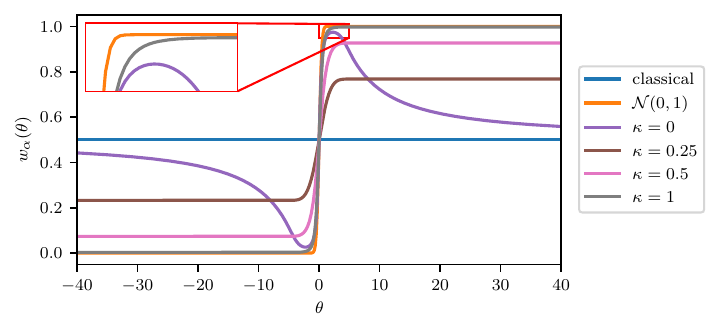}
  \caption{Weight function $w_\alpha(\theta)$ associated with the FAB-CRs under the Gaussian prior and other priors with bounded and vanishing influence. The horizontal line at $1/2$ represents the implicit $w_\alpha(\theta)$ for the $z$-interval.}
  \label{fig:fab_gaussian_w}
\end{figure}
By providing a general framework for obtaining uniformly bounded CRs, \cref{thm:robustCR} generalizes some previous attempts in this direction, such as the one recalled below.
\begin{example}
  \citet{Farchione2008} proposed to use the improper prior $\pi(d\theta)=\gamma d\theta +\delta_0$, with $\gamma>0$, within the FAB framework, showing that such a choice leads to robust FAB-CRs. Indeed, the marginal likelihood under $\pi$ is given by $f(y)= \gamma+ f_0(y)$, with $f(y)\to \gamma$ as $|y|\to\infty$, thereby satisfying assumption \eqref{eq:regvar1} for $\kappa = \delta = 0$.
\end{example}

We conclude with an interpretation of the limiting intervals in \cref{eq:lim_infCR,eq:lim_supCR} as FAB-CRs under specific improper priors; see also \cref{sec:exp_tilt_limits}.
\begin{remark}
  For $\xi\in\bbR$, the improper exponential-tilt prior
  \[
    \pi_\xi(d\theta)=\exp\left(-\frac{\xi\theta}{\sigma}\right)\,d\theta
  \]
  yields a constant-spending FAB procedure with $w_{\alpha,\xi}=g_\alpha^{-1}(2\xi)$, focal point $y-\sigma\xi$, and
  \[
    \CR^{(\xi)}(y)
    =
    \left[
    y-\sigma \Phi^{-1}(1-\alpha(1-w_{\alpha,\xi})),
    \,
    y+\sigma \Phi^{-1}(1-\alpha w_{\alpha,\xi})
    \right],
  \]
  which reduces to the standard $z$-interval when $\xi=0$ (flat improper prior). 
  The limiting intervals in \cref{thm:robustCR} are precisely these constant-spending FAB intervals: for $\kappa=0$ they reduce to the standard $z$-interval, while for $\kappa>0$,
  \[
    \lim_{y\to\infty}\CR(y)-y=\CR^{(\kappa)}(0),
    \qquad
    \lim_{y\to-\infty}\CR(y)-y=\CR^{(-\kappa)}(0).
  \]
  Thus, under prior--data conflict, the robust FAB procedure asymptotically behaves like a constant-spending one generated by the one-sided exponential tilt matching the relevant tail.
\end{remark}

%% Robust FAB with shrinkage priors
\section{Shrinkage priors for robust FAB confidence regions}
\label{sec:shrinkagepriors}
We now discuss shrinkage priors for the Gaussian likelihood model \eqref{eq:gaussian_likelihood} that yield robust and tractable FAB-CRs.
These priors shrink the FAB focal point and confidence region toward zero when the data are compatible with sparsity, while their tails prevent the unbounded-volume behavior seen under Gaussian priors.
To achieve this, we focus on priors whose marginal likelihood $f(y)$ satisfies the tail conditions of \cref{thm:robustCR} and is available in closed form, so that the weight function $w_\alpha(\theta_0)$ can be computed efficiently via \cref{eq:boundarycondition_1_w}.

We start by recording two useful invariance properties of the FAB weight function $w_\alpha(\theta_0)$.
\begin{proposition}\label{prop:w_properties}
    Assume that $\pi(d\theta)$ satisfies \cref{assumpt:priorpi}.
    Additionally, assume that $\pi(d\theta)$ depends on $\sigma$ as a scale parameter, i.e., $\pi(d\theta) = \pi(d\theta; \sigma)$ with  $\pi(A; \sigma) = \pi(A/\sigma; 1)$ for all $A \in \mathcal{B}(\mathbb{R})$, where $A/\sigma := \{x/\sigma : x \in A\}$.
    Define $w_\alpha(\theta_0; \sigma)$ as in \cref{eq:wfunction} for prior $\pi$.
    Then, for any $\theta_0 \in \mathbb{R}$,
    \begin{equation*}
        w_\alpha(\theta_0; \sigma) = w_\alpha\left(\theta_0/\sigma; 1 \right).
    \end{equation*}
    Moreover, if $\pi(d\theta; 1)$ is symmetric, i.e., $\pi(A;1) = \pi(-A;1)\ \forall A \in \mathcal{B}(\mathbb{R})$, then $w_\alpha(-\theta_0; 1) = 1 - w_\alpha(\theta_0; 1)$.
\end{proposition}
Throughout this section, when $\pi(d\theta)$ is absolutely continuous w.r.t.~the Lebesgue measure, we denote its density by $\pi(\theta)$, with some abuse of notation.
In general, obtaining the FAB-CR $\CR(y)$ requires evaluating the weight function $w_\alpha(\theta_0;\sigma)$ for a grid of values of $\theta_0$.
In practice, \cref{prop:w_properties} often allow us to compute $w_\alpha(\theta_0;\sigma)$ only for $\theta_0 > 0$ and $\sigma = 1$.

We focus on priors defined as scale mixtures of normals \citep{Andrews1974},
\begin{equation}
    \pi(\theta) = \int_0^\infty \mathcal{N}(\theta; 0, \sigma^2 \tau^2) d F_{\tau^2}(\tau^2), \label{eq:gaussian_scale_mixture}
\end{equation}
where $F_{\tau^2}(\tau^2)$ is a CDF on $[0, \infty)$ that does not depend on $\sigma$.
Such priors are symmetric, satisfy the scale property in \cref{prop:w_properties}, and encompass several priors used in the literature.
By defining $\tau^2$ as a discrete random variable taking values $0$ and $\tau_0^2$, one recovers the spike and slab prior used by \citet{Hoff2019}, whose marginal likelihood $f(y)$ is tractable, but does not satisfy the assumptions of \cref{thm:robustCR}.
Instead, taking $\tau^2$ to be an inverse-gamma random variable induces a Student-$t$ prior, whose marginal likelihood $f(y)$ has power-law tails, leading to robust FAB-CRs, but must be computed numerically.
Below, we identify scale mixtures that combine shrinkage, robustness, and closed-form marginal likelihoods.
\Cref{sec:vanishing_influence} treats priors with vanishing influence ($\kappa=0$), whereas \cref{sec:bounded_influence} discusses a prior with bounded, non-vanishing influence ($\kappa>0$).
Although these priors appear in the literature on robust Bayesian inference, none of them has been previously studied in the context of FAB.

\subsection{Priors with vanishing influence}
\label{sec:vanishing_influence}

For scale mixtures of normals \eqref{eq:gaussian_scale_mixture}, power law tails of the survival function $\bar F_{\tau^2}(t) := 1 - F_{\tau^2}(t)$ transfer directly the marginal likelihood $f(y)$.
The following proposition makes this precise.

\begin{proposition}\label{prop:gcm_power_law}
    Let $\pi(\theta)$ be a scale mixture of normals \eqref{eq:gaussian_scale_mixture} with mixing CDF $F_{\tau^2}$.
    If
    \begin{equation}
        \bar F_{\tau^2}(t) = 1 - F_{\tau^2}(t) \sim \frac{C_1}{\beta}\, t^{-\beta}
        \qquad \text{as }t \to \infty,
        \label{eq:tau2_tail_rv}
    \end{equation}
    for some constants $\beta, C_1 > 0$, then the marginal likelihood $f(y)$ has power-law tails with exponent $2 \beta + 1$, i.e., for some constant $C_2 > 0$,
    \begin{equation*}
        f(y) \sim C_2 |y|^{-(2 \beta + 1)} \quad \text{as} \quad y \to \pm \infty.
    \end{equation*}
\end{proposition}
Thus, any scale mixture \eqref{eq:gaussian_scale_mixture} satisfying \cref{eq:tau2_tail_rv} satisfies \cref{eq:regvar1} with $\kappa = 0$ and $\delta = 2\beta + 1$, thereby leading to robust FAB-CRs through \cref{thm:robustCR}.
If $F_{\tau^2}$ has density $f_{\tau^2}$ w.r.t.~the Lebesgue measure, the simpler condition
\begin{equation*}
    f_{\tau^2}(t) \sim C_1 t^{-(\beta + 1)} \quad \text{as} \quad t \to \infty
\end{equation*}
implies \cref{eq:tau2_tail_rv} by Karamata's theorem \citep[Proposition 1.5.10]{Bingham1989}.

In order to take advantage of \cref{prop:gcm_power_law}, we specify a beta prime distribution with parameters $a,b > 0$ for $\tau^2$, denoted by $\tau^2 \sim \mathrm{BP}(a, b)$, whose density is given by
\begin{equation}
    f_{\tau^2}(\tau^2) = \frac{(\tau^2)^{b-1} (1 + \tau^2)^{-(a + b)}}{B(a, b)},  \label{eq:beta_prime_density}
\end{equation}
where $B(a, b)$ is the beta function.
The beta prime family is well-known in the shrinkage literature \citep{Carvalho2010,Armagan2011,Polson2012}.
It is easy to see that its density behaves as a polynomial both at infinity and at zero, i.e.,
\begin{align}
    f_{\tau^2}(\tau^2) &\sim \frac{1}{B(a, b)}\frac{1}{(\tau^2)^{a + 1}} \quad \text{as} \quad \tau^2 \to +\infty, \label{eq:betaprime_poly_infty} \\
    f_{\tau^2}(\tau^2) &\sim \frac{1}{B(a, b)}(\tau^2)^{b - 1} \quad \text {as} \quad \tau^2 \to 0. \label{eq:betaprime_poly_zero}
\end{align}
The power-law tails \eqref{eq:betaprime_poly_infty} guarantee the robustness of the resulting FAB-CR through \cref{prop:gcm_power_law}, while the polynomial behavior at zero \eqref{eq:betaprime_poly_zero}, controlled by parameter $b$, relates to the strength of shrinkage toward zero.
Contrary to common choices of $f_{\tau^2}(\tau^2)$ with power-law tails, the beta prime density \eqref{eq:beta_prime_density} induces closed-form expressions for the marginal likelihood $f(y)$ and, in turn, for the FAB focal point $\widehat{\theta}^\text{FAB}(y)$ in \cref{eq:posteriormean}.
\begin{proposition}\label{prop:betaprime_marginal}
    Let $\pi(\theta)$ be a scale mixture of normals \eqref{eq:gaussian_scale_mixture} with beta prime mixing density \eqref{eq:beta_prime_density}, for $a, b > 0$.
    Then, the marginal likelihood $f(y)$ is given by
    \begin{equation*}
        f(y) = \frac{1}{\sqrt{2 \pi \sigma^2}} \times \frac{\Gamma(a + 1/2) \Gamma(a + b)}{\Gamma(a) \Gamma(a + b + 1/2)} \times \OneFOne\left(a + \frac{1}{2}, a + b + \frac{1}{2}, -\frac{y^2}{2 \sigma^2} \right),
    \end{equation*}
    where $\OneFOne(\alpha, \beta, z)$ is Kummer's confluent hypergeometric function of the first kind.
    Furthermore, the associated focal point $\widehat{\theta}^\text{FAB}(y) = \mathbb{E}[\theta | y]$ is given by
    \begin{equation*}
        \widehat{\theta}^\FAB(y) = y \times \left(1 - \left(\frac{a + \frac{1}{2}}{a + b + \frac{1}{2}}\right)\frac{\OneFOne\left(a + \frac{3}{2}, a + b + \frac{3}{2}, -\frac{y^2}{2 \sigma^2} \right)}{\OneFOne\left(a + \frac{1}{2}, a + b + \frac{1}{2}, -\frac{y^2}{2 \sigma^2} \right)}\right).
    \end{equation*}
\end{proposition}

Several choices of $(a, b)$ yield simpler expressions for $f(y)$ and recover known shrinkage priors.
We discuss some of these choices below, while a brief overview of the special functions involved in the expressions of $f(y)$ and $\widehat{\theta}^\text{FAB}(y)$ is provided in \cref{app:specialfunctions}.
\begin{example}\label{ex:horseshoe}
    Taking $a=b=1/2$ gives the half-Cauchy density on $\tau$ and hence the horseshoe prior $\pi(\theta)$ and estimator $\widehat{\theta}^\text{FAB}(y)$ \citep{Carvalho2010}.
    In this case,
    \begin{equation*}
        f(y)
        =
        \frac{2}{\pi^{3/2}}
        \frac{1}{|y|}
        D\left(\frac{|y|}{\sqrt{2\sigma^2}}\right),
    \end{equation*}
    where $D(z)=e^{-z^2}\int_0^z e^{t^2}dt$ is Dawson's function.
\end{example}
\begin{example}\label{ex:pareto}
    Taking $a = 1/2$ and $b = 1$ gives $f_{\tau^2}(\tau^2) = \tfrac{1}{2}(1 + \tau^2)^{-3/2}$, which corresponds to a generalized Pareto distribution on $z = \tau^2/2$ \citep{Griffin2011}.
    In this case,
    \begin{equation*}
        f(y) = \begin{cases}
            \frac{\sigma}{\sqrt{2 \pi}} \frac{1 - \exp(-y^2/(2 \sigma^2))}{y^2} & y \neq 0 \\
            \frac{1}{2\sigma \sqrt{2 \pi}} & y = 0.
        \end{cases}
    \end{equation*}
\end{example}
\begin{example}\label{ex:modified_bessel}
    Taking $a = 1$ and $b = 1/2$ gives $f_{\tau^2}(\tau^2) = \tfrac{1}{2} (\tau^2)^{-1/2} (1 + \tau^2)^{-3/2}$.
    In this case,
    \begin{equation*}
        f(y) = \frac{1}{\sqrt{2 \pi \sigma^2}} \times \frac{\pi}{4} \times \exp\left(-\frac{y^2}{4 \sigma^2}\right) \times \left[I_0\left(\frac{y^2}{4 \sigma^2}\right) - I_1\left(\frac{y^2}{4 \sigma^2}\right)\right],
    \end{equation*}
    where $I_n(z)$ is the $n$-th order modified Bessel function of the first kind.
\end{example}
\Cref{fig:fab_gaussian_vanishing} compares the FAB-CRs induced by these three beta prime mixtures.
\begin{figure}[t]
    \centering
    \includegraphics[width=.9\textwidth]{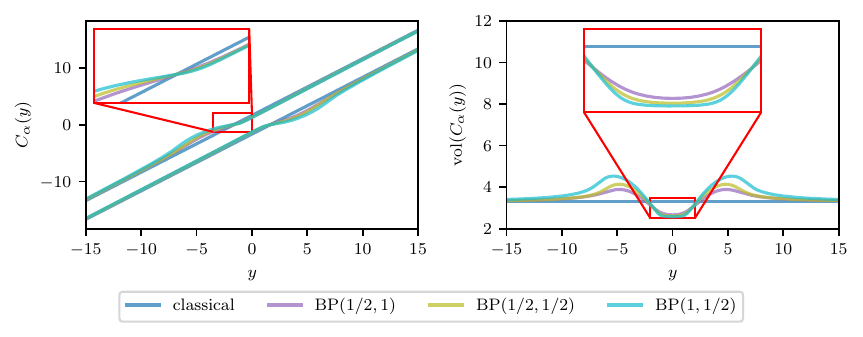}
    \caption{Comparison of standard and FAB procedures under the prior \eqref{eq:gaussian_scale_mixture} with mixing measure $\mathrm{BP}(1/2,1/2)$, $\mathrm{BP}(1/2,1)$ and $\mathrm{BP}(1,1/2)$, when $Y\mid \theta\sim \Normal(\theta, 1)$ for $\alpha = 0.1$.}
    \label{fig:fab_gaussian_vanishing}
\end{figure}
All three FAB-CRs are shorter than the standard CI when $y$ is close to zero.
Unlike the Gaussian-prior FAB-CR in \cref{fig:fab_gaussian_gaussian}, however, they are uniformly bounded and revert to the standard $z$-interval as $|y|$ grows.
The corresponding $p$-value functions, shown in \cref{fig:fab_gaussian_vanishing_cc}, highlight the same behavior and illustrate that the posterior mean coincides with the FAB focal point.

Finally, while we do not pursue this direction here, it is worth noting that spike and slab priors with a spike at $0$ and a slab component with power-law tails, such as a Cauchy-tailed slab \citep{Johnstone2004}, also satisfy the assumptions of \cref{prop:gcm_power_law}: the atom at $0$ encodes exact prior sparsity, while the heavy-tailed slab yields robust FAB-CRs. \looseness=-1

\subsection{Prior with bounded influence}\label{sec:bounded_influence}
As mixing distribution $F_{\tau^2}(\tau^2)$ for the scale mixture of normals prior \eqref{eq:gaussian_scale_mixture}, consider the CDF of an exponential random variable with rate $\kappa^2/2$ for $\kappa > 0$, denoted by $\tau^2 \sim \mathrm{Exp}(\kappa^2/2)$.
The resulting prior is the Laplace, or double-exponential, density
\begin{equation*}
    \pi(\theta;\sigma)=\frac{\kappa}{2\sigma} \exp\left(-\frac{\kappa}{\sigma}|\theta|\right),
\end{equation*}
with scale parameter $\sigma/\kappa$ \citep{Pericchi1992}.
Its marginal likelihood is
\begin{equation*}
    f(y) = \frac{\kappa}{2\sigma}\exp\left(\frac{\kappa^2}{2}\right)\left[\exp\left(-\frac{\kappa y}{\sigma}\right)\Phi\left(\frac{y-\sigma \kappa}{\sigma}\right)+\exp\left(\frac{\kappa y}{\sigma}\right)\Phi\left(-\frac{y+\sigma\kappa}{\sigma}\right)  \right],
\end{equation*}
which implies
$$
    f(y)\sim \frac{\kappa}{2\sigma}\exp\left(\frac{\kappa^2}{2}\right)\exp\left(-\frac{\kappa |y|}{\sigma}\right) \quad \text{as } y\to \pm\infty.
$$
Thus, $f(y)$ satisfies \cref{eq:regvar1} with $\kappa > 0$ and $\delta=0$.
The corresponding focal point is
\begin{equation*}
    \widehat\theta^\FAB(y)=\xi(y)(y+\sigma\kappa)+(1-\xi(y))(y-\sigma\kappa)=y\times \left( 1 + \frac{\sigma\kappa}{y}(2\xi(y)-1) \right),
\end{equation*}
where $\xi(y)=\left(1+\exp\left(-2\kappa y/\sigma\right)\Phi(\frac{y-\sigma\kappa}{\sigma})/\Phi(-\frac{y+\sigma\kappa}{\sigma})   \right)^{-1}$.
By \cref{thm:robustCR}, this prior has bounded influence.
That is, for $y\to\infty$,
\begin{align*}
    \widehat\theta^\FAB(y)-y &\to  -\sigma\kappa\qquad \text{and}\qquad
    C_\alpha(y)-y &\to \left[- \sigma \Phi^{-1}\left(1-\alpha c_\alpha\right), \sigma \Phi^{-1}\left(1-\alpha (1 - c_\alpha)\right)\right],
\end{align*}
where $c_\alpha = g_\alpha^{-1}(-2\kappa) \in (0, 1/2)$.
Therefore, under the Laplace prior, the FAB-CR remains uniformly bounded but does not revert to the standard $z$-interval, and the corresponding FAB estimator does not converge to the MLE.
This behavior is illustrated in \cref{fig:fab_gaussian_bounded}.
\begin{figure}
    \centering
    \includegraphics[width=\textwidth]{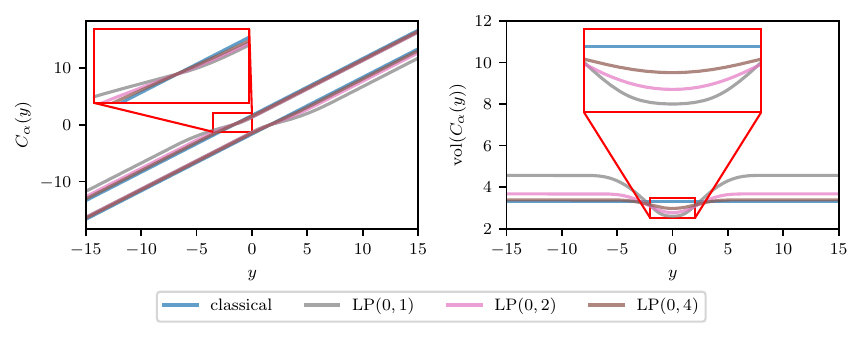}
    \caption{Comparison of the standard and FAB procedures under the Laplace prior $\mathrm{LP}(0, \sigma / \kappa)$ with scale parameter $\sigma/\kappa \in \{1, 2, 4\}$, when $Y\mid \theta\sim \Normal(\theta, 1)$ for $\alpha = 0.1$.}
    \label{fig:fab_gaussian_bounded}
\end{figure}

The Laplace prior also helps clarify the relationship between FAB-CRs, Bayesian credible regions, and the corresponding focal points, which generally differ.
The focal point of the Bayesian $(1-\alpha)$ highest posterior density (HPD) region under the Laplace prior is the maximum a posteriori (MAP) estimator, which is given by
$$
    \widehat\theta^\MAP(y) =
    \left\{
    \begin{array}{ll}
        0 & |y|\leq \sigma\kappa \\
        y- \sigma\kappa \times \text{sign}(y) & |y| > \sigma\kappa
    \end{array}
    \right.
$$
By contrast, the FAB focal point is the posterior mean.
Thus, $\widehat\theta^\FAB(y)$ and $\widehat\theta^\MAP(y)$ generally differ for finite $y$.
The corresponding $p$-value functions and focal point comparisons are reported in \cref{app:bounded_influence}.
Lastly, as shown by \citet{Pericchi1992}, we have $\theta\mid y \overset{d}{\simeq}\Normal(y-\sigma\kappa\times \text{sign}(y), \sigma^2)$ for large $|y|$.
Therefore, although $\widehat\theta^\FAB(y)$ and $\widehat\theta^\MAP(y)$ coincide asymptotically, the FAB-CR and the Bayesian HPD region remain different.

\subsection{Summary and practical guidelines}
\label{sec:practical_guidelines}

\Cref{thm:robustCR} shows the role of the tails of the marginal likelihood $f(y)$ in determining the behavior of FAB-CRs under prior--data conflict.
Gaussian priors yield unbounded FAB-CRs; bounded-influence priors yield uniformly bounded FAB-CRs with asymptotically shifted focal points; and vanishing-influence priors yield FAB-CRs that revert to the $z$-interval.
\cref{tab:priors_summary} summarizes the properties of the priors discussed so far in the context of \cref{thm:robustCR}.
\begin{table}
    \centering
    \caption{Summary of priors discussed in \cref{sec:gaussianlikelihood,sec:shrinkagepriors} in the context of \cref{thm:robustCR}.}
    \small
    \setlength{\tabcolsep}{5pt}
    \label{tab:priors_summary}
    \begin{tabular}{@{}p{1.8cm}p{3.5cm}p{2.2cm}p{3.1cm}p{4.1cm}@{}}
        \toprule
        \textbf{Prior} & \textbf{Definition} & \textbf{Parameters} & \textbf{Influence regime} & \textbf{Tractability of $f(y)$} \\
        \midrule
        Gaussian & $\theta\sim \mathcal N(0,\tau^2)$ & $\tau > 0$ & Unbounded (outside Thm~\ref{thm:robustCR}) & Closed form (elementary) \\
        Laplace & Scale mixture \eqref{eq:gaussian_scale_mixture} with $\tau^2\sim \mathrm{Exp}(\kappa^2/2)$ & $\kappa > 0$ & Bounded ($\kappa > 0$) & Closed form (elementary) \\
        Beta prime mixture & Scale mixture \eqref{eq:gaussian_scale_mixture} with $\tau^2\sim \mathrm{BP}(a,b)$ & $a, b > 0$ & Vanishing ($\kappa = 0$) &
        Closed form ($\OneFOne$ with possible simplifications) \\
        \bottomrule
    \end{tabular}
\end{table}

All three shrinkage regimes can reduce FAB-CR volume near zero, at the cost of larger regions when the data conflict with the prior.
Such a trade-off is inherent to the FAB construction and can be controlled by tuning the prior scale via prior--specific parameters.
For the Gaussian prior, $\tau$ directly controls the prior scale and, hence, the amount of shrinkage toward zero (see \cref{fig:fab_gaussian_gaussian,fig:fab_gaussian_gaussian_cc}).
For the Laplace prior, $\kappa$ controls the prior scale via $\sigma/\kappa$, but also affects the asymptotic shift of the resulting FAB-CR (see \cref{fig:fab_gaussian_bounded,fig:fab_gaussian_bounded_cc}).
For beta prime mixtures, smaller $b$ increases shrinkage near zero through \eqref{eq:betaprime_poly_zero}, while $a$ controls tail mass through \eqref{eq:betaprime_poly_infty} and hence the rate at which the FAB-CR approaches the $z$-interval (see \cref{fig:fab_gaussian_vanishing,fig:fab_gaussian_vanishing_cc}).
Alternatively, one could introduce an additional global scale parameter $s$ multiplying $\sigma$ in the mixture representation \eqref{eq:gaussian_scale_mixture} to control the prior scale more directly, but this generally forfeits the closed-form expressions for $f(y)$ and $\widehat\theta^\FAB(y)$ in \cref{prop:betaprime_marginal}.
In all cases, the prior--specific parameters can be tuned by quantile matching in the presence of domain knowledge, or chosen adaptively when statistically independent estimators are available \citep{Hoff2019}.
In the absence of such information, the horseshoe prior, corresponding to the beta prime mixture with $a=b=1/2$, is a reasonable default: it combines strong shrinkage near zero with heavy tails, and the specific parameterization \eqref{eq:gaussian_scale_mixture} using the data scale $\sigma$ as the sole scale parameter is recommended as a way to avoid under- or over-shrinkage \citep{Piironen2017}.

%% Extension to exponential family likelihood
\section{Extension to natural exponential families}
\label{sec:exponentialfamily}
We now generalize \cref{thm:postmeaninCI} to likelihoods $f_\theta(y)$ belonging to an exponential family.
For concreteness, we use the case of a Poisson random variable with mean $\theta$, denoted $Y \sim \Poisson(\theta)$, as a running example.
We begin by recalling some background on natural exponential families (NEFs); see \citet{Brown1986,Wainwright2008} for additional details.
\begin{definition}[Natural exponential family]\label{def:NEF}
  Let $\nu$ denote a sigma-finite measure on the Borel subsets of $\bbR^d$, and consider $h \colon \bbR^d\to[0,\infty)$.
  Define the set
  \begin{align*}
    \naturalset = \left\{ \eta \mid \int_{\bbR^d} h(y)e^{\eta^\top y }\nu(dy)<\infty\right\}.
  \end{align*}
  For any $\eta\in \naturalset$, let
  \begin{align*}
    \psi(\eta) = \log\left(\int_{\bbR^d} h(y)e^{\eta^\top y }\nu(dy)\right),
  \end{align*}
  and define $f_\eta(y)= h(y)e^{\eta^\top y -\psi(\eta)}$.
  The class $\{f_\eta(y):\eta\in \naturalset\}$ is a $d$-dimensional natural exponential family (NEF) of probability densities with respect to $\nu$.
  The set $\naturalset$ and the function $\psi$ are called the natural parameter space and the cumulant function, respectively.
\end{definition}
We will work with regular, minimal NEFs.
That is, $\naturalset$ is open and, writing $\calY\subseteq\bbR^d$ for the support of $h(y)\nu(dy)$ and $\calM$ for the closure of the convex hull of $\calY$, we assume $\dim\naturalset=\dim\calM=d$.
Throughout this section, $\nu$ is either Lebesgue or counting measure.
In this case, for a given $\eta \in \naturalset$, the expectation of $Y$ can be expressed in terms of the cumulant function $\psi$.

\begin{proposition}[{\citealp[Section 3.5, Propositions 2 and 3]{Wainwright2008}}]
  \label{prop:nef_mean}
  The set $\naturalset$ is convex and the function $\psi$ is strictly convex on $\naturalset$.
  For any $\eta\in\naturalset$,
  \begin{align*}
    \E[Y\mid \eta] &= \nabla\psi(\eta).
  \end{align*}
  Moreover, the gradient mapping $\nabla\psi \colon \naturalset\to\calM$ is one-to-one.
\end{proposition}

For instance, the Poisson distribution with mean $\theta$ is a regular, minimal NEF with natural parameter $\eta = \log(\theta) \in \naturalset = \bbR$, cumulant function $\psi(\eta) = e^\eta$, $h(y) = 1/y!$, reference measure $\nu$ given by the counting measure on $\bbN$, and \cref{prop:nef_mean} implies that $\E[Y\mid \eta] = \nabla\psi(\eta) = e^\eta$.
\Cref{assumpt:priorpiNEF} is the analogue of \cref{assumpt:priorpi} for natural exponential families.

\begin{assumption}
  \label{assumpt:priorpiNEF}
  Let $\pi(d\eta)$ be a positive Radon measure on $\naturalset$, with full support, and such that $0<\int_{\naturalset} f_\eta(y)\pi(d\eta)<\infty$ for any $y\in\calY$.
  We refer to $\pi$ as the prior distribution on the natural parameter $\eta$.
  Assume that $\pi(d\eta)$ is absolutely continuous w.r.t.~a sigma-finite reference measure $\xi$ on $\bbR^d$, with density $g(\eta)$.
  Finally, assume that the following set is open: \looseness=-1
  \begin{equation*}
    \widetilde\calY = \left\{ y\in \bbR^d \mid \int_{\bbR^d} e^{\eta^\top y - \psi(\eta)} g(\eta) \xi(d\eta) <\infty \right\}.
  \end{equation*}
\end{assumption}
\begin{remark}
  As in \cref{sec:gaussianlikelihood}, $\pi$ and $\xi$ will usually be a proper probability distribution and the Lebesgue measure, respectively, but \cref{assumpt:priorpiNEF} also allows for improper priors and spike and slab priors.
  In practice, one may instead specify a prior on the mean $\nabla\psi(\eta)$, which induces a prior $\pi$ on $\eta$ via its push-forward measure w.r.t.~the inverse gradient map. \looseness=-1
\end{remark}
Note that $\calY \subseteq \widetilde\calY$, which implies $\dim(\widetilde\calY)=d$.
For any $y\in\widetilde\calY$, define
\begin{equation}
  \lambda(y)=\log \left(\int_{\bbR^d} e^{\eta^\top y - \psi(\eta)}  g(\eta)\xi(d\eta)\right) \quad\text{and}\quad \widetilde f_y(\eta) = \left[g(\eta)e^{ - \psi(\eta)}\right] e^{\eta^\top y - \lambda(y)},
  \label{eq:NEFposterior}
\end{equation}
for any $\eta\in\naturalset$.
Then, by definition, the class $\{\widetilde f_y(\eta) ; y\in\widetilde\calY \}$ is a regular, minimal, NEF of densities over $\naturalset$ w.r.t.~the reference measure $\xi$, and whose cumulant function is $\lambda(y)$.
\begin{remark}
  For $y\in\calY$, $\widetilde f_y(\eta)$ may be interpreted as the posterior density of $\eta$ given $y$ under likelihood $f_\eta(y)$ and prior $\pi$.
  Similarly, $\lambda(y)$ may be interpreted as $\log\left(\frac{f(y)}{h(y)}\right)$ where $f(y)=\int f_\eta(y)g(\eta)\xi(d\eta)$ is the marginal likelihood. The log-ratio between the marginal and the likelihood takes the form $\lambda_\eta(y) = \log \frac{f(y)}{f_\eta(y)}=\lambda(y)-\eta^\top y+ \psi(\eta)$.
  \label{rem:posterior_nef}
\end{remark}
For instance, in the Poisson example, the conjugate prior $\theta\sim \Gammadist(a, p/(1-p))$, with $a>0$ and $p\in(0,1)$, induces a log-gamma prior on $\eta$, which satisfies \cref{assumpt:priorpiNEF}, and thus a log-gamma posterior for $\eta$ given $y$.
This gives $\lambda(y) = \log \Gamma(a + y) - \log \Gamma(a) + y \log(1-p) + a \log(p)$, which is defined for any $y\in\widetilde\calY = (-a, \infty)$.

We now define the FAB-CR for the natural parameter and for any one-to-one statistical functional of a random variable in a NEF.
Compared with the general FAB construction in \cref{sec:fabbackground}, two modifications are needed:
\begin{enuminline}
  \item for discrete observations, exact coverage \eqref{eq:coverage} is replaced by a lower bound, and
  \item to derive the focal point, the log-likelihood ratio is extended to the larger (convex) set $\widetilde \calY \supseteq \calY$.
\end{enuminline}
\begin{definition}
  Let $\alpha\in(0,1)$ and $Y\sim f_\eta$ be a random variable in a regular, minimal NEF with cumulant $\psi$.
  Let $\theta=\statfunc(\eta)$ be a one-to-one statistical functional of $f_\eta$, e.g., its mean, median, or quantile.
  In the case of the mean, $\statfunc=\nabla\psi$.
  Let $g(\eta)$ be some prior density on the natural parameter $\eta$ w.r.t.~some reference measure $\xi$.
  For any $y\in\widetilde \calY$, consider the log-likelihood ratio $\lambda_\eta(y) = \lambda(y)-\eta^\top y+ \psi(\eta)$.
  The $(1-\alpha)$ FAB-CR for $\theta$ is given by
  \begin{align*}
    \CR(y)=\{\statfunc(\eta) \mid \eta\in \CRnp(y)\},
  \end{align*}
  where $\CRnp$ is the $(1-\alpha)$ FAB-CR for the natural parameter $\eta$, defined as
  $$
    \CRnp(y)=\{\eta\in\naturalset \mid y\in \ARnp(\eta)\}, \quad\text{with}\quad \ARnp(\eta)=\left\{y\in \widetilde \calY\mid \lambda_\eta(y) \leq k_\alpha(\eta)\right\},
  $$
  where $k_\alpha(\eta)$ is the smallest value such that
  $
    \int_{\ARnp(\eta)} f_\eta(y)\nu(dy)\geq 1-\alpha.
  $
  \label{def:fabcrdefnef}
\end{definition}
For discrete observations, $\lim_{\alpha\to 1}\vol(\CR(y))\neq 0$, and $\cap_{\alpha\in(0,1)}\CR$ is therefore an interval, preventing the use of \eqref{eq:focalpointFABCR} as a focal point.
Instead, we define
\begin{equation}
  \widehat\theta^\FAB(y)=\statfunc(\widehat\eta^\FAB(y)),
  \quad\text{where}\quad
  \widehat\eta^\FAB(y)=y_{\min}^{-1}(y), \label{eq:etaFAB_nef}
\end{equation}
and
$y_{\min}(\eta) = \arg\min_{y\in\widetilde\calY} \lambda_\eta(y).$
This definition reduces to \eqref{eq:focalpointFABCR} when $\lim_{\alpha\to 1}\vol(\CR(y))= 0$.

By \cref{rem:posterior_nef}, for $y\in \calY$, density \eqref{eq:NEFposterior} may be seen as the posterior density of $\eta$ under likelihood $f_\eta(y)$ and prior $\pi$.
Hence, by \cref{prop:nef_mean}, its posterior mean is given by
\begin{align}
  \widehat\eta(y) = \E[\eta \mid y] = \int \eta \widetilde f_y(\eta) \xi(d\eta)=\nabla\lambda(y), \label{eq:etapostmean_nef}
\end{align}
The next result extends \cref{thm:postmeaninCI} to NEFs: the FAB focal point \eqref{eq:etaFAB_nef} for $\eta$ coincides with the posterior mean \eqref{eq:etapostmean_nef}, which is thus always contained in the corresponding FAB-CR.
\begin{theorem}\label{thm:postmeaninCI_nef}
  Assume that \cref{assumpt:priorpiNEF} is satisfied, and let $\alpha\in(0,1)$.
  For any $\eta\in\naturalset$, $\ARnp(\eta)$ is convex, and $\widehat\eta^{-1}(\eta)\in\ARnp(\eta)$.
  Then, for any observation $y\in\calY$,
  \begin{align*}
    \widehat\eta^{\FAB}=\nabla\lambda(y)=\E[\eta\mid y]\in \CRnp(y)
  \end{align*}
  and $\widehat\theta^{\FAB}=\statfunc(\nabla\lambda(y))=\statfunc(\E[\eta\mid y])\in \CR(y).$
\end{theorem}
\begin{example}
  Continuing the running example, where $Y \sim \Poisson(\theta)$ with prior $\theta \sim \Gammadist(a, p/(1-p))$, we have
  $
    \widehat\eta^{\FAB}(y) = \E[\eta\mid y] = \lambda'(y) =\digamma(a+y)+\log(1-p) \in \CRnp(y),
  $
  where $\digamma(z)$ is the digamma function.
  Since $\theta = \E[Y | \eta] = \psi'(\eta) = e^\eta$, the FAB estimator for $\theta$ is given by $\widehat\theta^{\FAB}(y) = e^{\widehat\eta^{\FAB}(y)} \in \CR(y)$.
  \Cref{fig:fab_poisson} illustrates the FAB procedure for the Poisson likelihood.
  \Cref{fig:fab_poisson_cc} in \cref{app:additional_figures} reports the corresponding $p$-value functions.
\end{example}
\begin{figure}
  \centering
  \includegraphics[width=\textwidth]{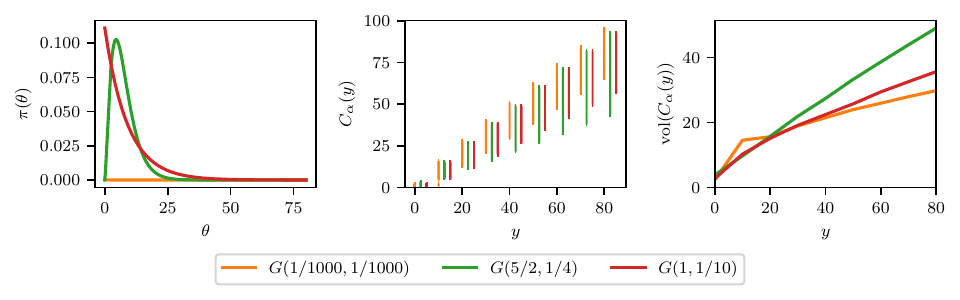}
  \caption{Comparison of the FAB procedure applied to the likelihood $Y \sim \mathrm{Poisson}(\theta)$ under a prior $\theta \sim \Gammadist(a, p / (1 - p))$ for different choices of $a$ and $p$, and where $\alpha = 0.1$.}
  \label{fig:fab_poisson}
\end{figure}

The same construction applies to other discrete NEFs.
For example, the binomial--beta case yields a FAB estimator for the logit parameter equal to the posterior mean under the corresponding logistic-beta posterior; under the uniform beta prior, the resulting FAB-CRs recover the binomial regions of \citet{Sterne1954}.
The multinomial--Dirichlet case is analogous, with the uniform Dirichlet prior recovering the regions of \citet{Chafai2009}.
Details and illustrations for these additional examples are given in \cref{app:additional_figures}.

%% Application to linear regression
\section{Application to linear regression}
\label{sec:linearregression}
To illustrate the behavior of the FAB priors proposed in \cref{sec:shrinkagepriors} in a common estimation problem, consider the linear regression model
\begin{align}
    \bY\mid \bbeta \sim \Normal(\bX \bbeta, \Sigma),
    \label{eq:likelihoodlinearregression}
\end{align}
where $\bY\in\bbR^n$,
$\bX$ is a fixed $n$--by--$p$ design matrix with full rank, $\bbeta=(\beta_1,\ldots,\beta_p)^\top\in\bbR^p$ is the vector of unknown coefficients,
and $\Sigma$ is a known positive definite matrix.
In many applications, $\Sigma$ describes i.i.d.~observation noise, i.e., $\Sigma = \sigma_Y^2 \mathbf{I}$ with $\sigma_Y > 0$.
For a given $\bx\in\bbR^p$, we aim to construct a valid $(1-\alpha)$ confidence region $C_{\alpha}(\bY;\bx)$ for $\bx^\top\bbeta\in\bbR$, that is
\begin{align*}
\Pr(\bx^\top\bbeta\in C_{\alpha}(\bY;\bx)\mid \bbeta) = 1-\alpha
\end{align*}
for any $\bbeta\in\bbR^p$.
This includes, as a special case, confidence regions for individual coefficients $\beta_j$, $j=1,\ldots,p$, obtained by choosing $\bx$ as a one-hot vector.

Let $\widehat\bbeta= (\bX^\top\Sigma^{-1}\bX)^{-1}\bX^\top \Sigma^{-1} \bY$ be the MLE of $\bbeta$.
Under the Gaussian model \eqref{eq:likelihoodlinearregression},
$\widehat\bbeta \mid \bbeta\sim \Normal(\bbeta, \widetilde \Sigma)$, where $\widetilde \Sigma=(\bX^\top\Sigma^{-1}\bX)^{-1}$.
Therefore, for any $\bx\in\bbR^p$,
\begin{equation*}
    \bx^\top \widehat\bbeta \mid \bbeta \sim \Normal(\bx^\top \bbeta, \bx^\top \widetilde \Sigma\bx ),
\end{equation*}
which corresponds to the setting of \cref{sec:gaussianlikelihood} with $\theta=\bx^\top \bbeta$ and $\sigma^2= \bx^\top \widetilde \Sigma\bx$.
As a result of this, the standard $z$-interval, given here by $(\bx^\top \widehat\bbeta\pm z_{1-\alpha/2}\sqrt{\bx^\top \widetilde \Sigma\bx})$, is a valid CI for the problem at hand, that always contains the MLE $\bx^\top \widehat\bbeta$.
Alternatively, in the presence of prior information on the likely location of $\bbeta$, one may wish to apply FAB under a prior $\pi(d\theta)$.
By \cref{thm:postmeaninCI}, if $\pi$  satisfies \cref{assumpt:priorpi},
then $\widehat\theta^\FAB(\bx^\top\widehat\bbeta) = \bbE[\theta \mid \bx^\top\widehat\bbeta]$,
is the focal point of this confidence region.
Additionally, by using the shrinkage priors discussed in \cref{sec:shrinkagepriors}, sparsity in $\bbeta$ may be encouraged, while also ensuring that FAB-CRs for $\theta$ remain uniformly bounded via \cref{thm:robustCR}.
In particular, if $\pi$ is of the scale mixture form \eqref{eq:gaussian_scale_mixture} with power-law tails, the volume of the confidence region converges to $2 z_{1-\alpha/2}\sqrt{\bx^\top \widetilde \Sigma\bx}$, which corresponds to the width of the standard confidence interval, as $|\bx^\top \widehat\bbeta|\to\infty$.
These properties are highlighted in the following simulation study. \looseness=-1

\textbf{Synthetic data.}
We sample the entries of $\bbeta \in \mathbb{R}^p$ i.i.d.~from a Gaussian distribution with mean zero and variance $\sigma_\beta^2$, and each row of $\bX \in \mathbb{R}^{n \times p}$ i.i.d.~from a multivariate normal distribution with mean zero, unit variance, and correlations $\mathrm{corr}(x_{ij}, x_{ik}) = 0.5^{|j-k|}$.
Then, for a given value of $\bbeta$ and $\bX$, we sample $\bY \in \mathbb{R}^n$ from the model \eqref{eq:likelihoodlinearregression} with $\Sigma = \sigma_Y^2 \mathbf{I}$, where $\sigma_Y = 10$.
In this setup, we are interested in obtaining marginal confidence regions for each of the scalar parameters $\beta_j$, $j=1,\ldots,p$.
In order to evaluate the performance of different FAB priors in the presence of disagreement with the data, for each $j$, we fix the prior location and scale parameters to zero and $\sqrt{\widetilde \Sigma_{jj}}$, respectively, and gradually increase $\sigma_\beta^2$.
\Cref{fig:fab_ols} compares the FAB-CRs obtained under the Gaussian ($\mathrm{G}$), horseshoe ($\mathrm{HS}$), and Laplace ($\mathrm{LP}$) priors against the standard $z$-interval in terms of average volume and coverage as a function of $\log(\sigma_\beta)$.
\begin{figure}
    \centering
    \includegraphics[width=\textwidth]{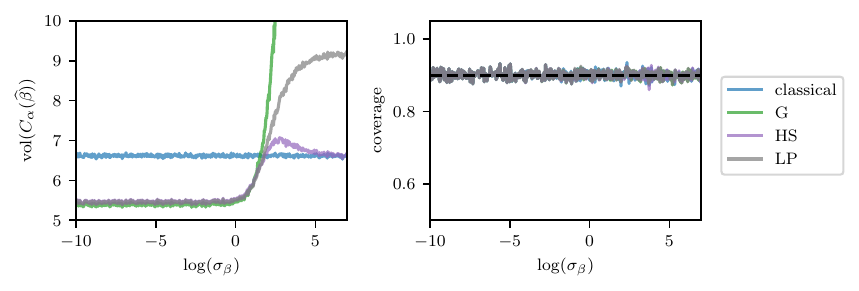}
    \caption{Linear regression simulation study. Comparison of the $z$-interval and FAB-CRs under different priors for $n=50$, $p=10$, $\alpha = 0.1$. Results are averaged over $100$ repetitions.}
    \label{fig:fab_ols}
\end{figure}
As expected, for small values of $\sigma_\beta^2$, FAB results in volume gains compared to the standard method.
However, as $\sigma_\beta^2$ grows, the zero-centered FAB priors become increasingly misspecified, and the FAB-CRs become larger than the classical CI.
Nonetheless, for very large values of $\sigma_\beta^2$, the behavior of the three priors differs: the width of the Gaussian CR diverges, while the Laplace and horseshoe priors yield uniformly bounded regions, with the latter eventually reverting to the standard $z$-interval.
In all cases, the coverage remains exact.
\Cref{fig:fab_ols_baselines} in \cref{app:additional_figures} reports the same experiment with additional baselines: Bayesian credible regions under the same priors, and desparsified Lasso confidence intervals \citep{Zhang2013}, with the true $\sigma_Y$ supplied to all methods for comparability.
The credible regions are narrower than the $z$-interval over the range of $\sigma_\beta^2$ considered: the Gaussian credible region has constant width because its posterior variance does not depend on the observed data $\mathbf{y}$, whereas the Laplace and horseshoe credible regions widen with $\sigma_\beta^2$ and approach the volume of the classical interval.
The desparsified Lasso interval is also shorter than the $z$-interval and has width largely insensitive to $\sigma_\beta^2$, since its width is determined mainly by the design and the supplied observation standard deviation rather than by the signal magnitude.
Crucially, none of these baselines maintains nominal coverage uniformly over $\sigma_\beta^2$: the credible regions over-cover when coefficients are near zero and can under-cover under substantial prior--data conflict, and the desparsified Lasso interval tends to under-cover when sparsity/regularity assumptions are violated.
Among the methods considered here, FAB-CRs are the only ones that simultaneously exploit shrinkage when present while maintaining nominal coverage by construction.

\textbf{Real-world data.} As a real-world example, we consider the breast cancer dataset analyzed by \citet{Shimomura2016} and publicly available on the GEO database \citep[GEO accession GSE73002]{Edgar2002}.
For $n = 4113$ human subjects tested for breast cancer, the dataset contains normalized microarray expression values for $2540$ microRNAs as covariates, and a binary label for breast cancer status as the response.
The covariate matrix contains a very small number of missing values ($\simeq 0.03\%$), which we impute by the median before standardizing each column to have mean zero and unit variance.
We aim to construct marginal confidence regions for the parameters $\bbeta$ of a linear regression model \eqref{eq:likelihoodlinearregression}, where the design matrix $\bX \in \mathbb{R}^{n \times p}$ contains the microRNA expression profiles and a column of ones for the intercept ($p = 2541$), the response variable $\bY \in \mathbb{R}^n$ contains the binary labels for breast cancer status, and $\Sigma = \sigma_Y^2 \mathbf{I}$.
As above, we encode sparsity into the FAB-CR procedure by specifying zero-mean independent priors on the entries of $\bbeta$, with scale parameter set to $\sqrt{\widetilde \Sigma_{jj}}$.
In this case, the observation variance $\sigma_Y^2$ required to compute $\widetilde \Sigma_{jj}$ is unknown, and we estimate it using the estimator
$\widehat\sigma_Y^2 = \sum_{i=1}^n (y_i - \mathbf{x}_i^\top \widehat\bbeta)^2/(n-p)$, where $\widehat\bbeta$ is the MLE of $\bbeta$.
Because the binary response $\bY$ is analyzed under the continuous working model \eqref{eq:likelihoodlinearregression} and $\sigma_Y^2$ is estimated rather than known, the nominal coverage of both the FAB-CR discussed above and the standard $z$-interval should be interpreted as approximate rather than exact.
This example is therefore intended primarily to illustrate the relative shrinkage and robustness behavior of the different FAB priors under a common working model.
\Cref{fig:fab_ols_real} compares the FAB-CRs obtained under the Gaussian, horseshoe, and Laplace priors against the standard $z$-interval in terms of volume.
\begin{figure}
    \centering
    \includegraphics[width=\textwidth]{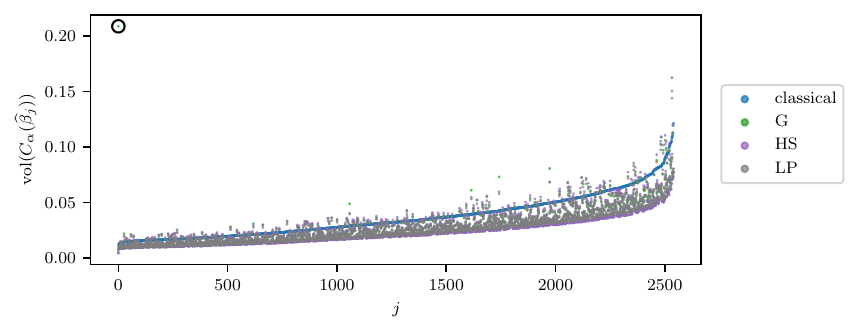}
    \caption{Linear regression real-world study. Comparison of the $z$-interval and FAB-CRs under different priors on the breast cancer dataset for $\alpha = 0.3$. Coefficients are ordered by width of the $z$-interval.}
    \label{fig:fab_ols_real}
\end{figure}
For this dataset, sparsity is very pronounced, and the FAB-CRs achieve smaller volume than the standard $z$-interval for most of the coefficients.
In particular, the Gaussian and Laplace priors outperform the standard $z$-interval for $\approx 93\%$ of the coefficients, while the horseshoe FAB-CRs are slightly more conservative, performing better for $\approx 90\%$ of the coefficients.
However, in cases where the priors are clearly misspecified, the heavier tails of the Laplace and horseshoe priors allow them to outperform the Gaussian prior.
This is particularly evident for the intercept coefficient, for which the Gaussian FAB-CR, indicated by the circle in the top left corner of \cref{fig:fab_ols_real}, attains a volume $\approx 48$ times larger than the corresponding $z$-interval.
In contrast, the Laplace CR remains bounded to a much smaller volume of $\approx 1.5$ times the $z$-interval, while the horseshoe CR reverts to the latter, highlighting its robustness.
For this experiment, realizations of the FAB-CRs and $z$-intervals for four representative coefficients are shown in \cref{fig:fab_ols_real_examples}.

%% Discussion
\section{Discussion}
\label{sec:discussion}
Pratt's \textit{Frequentist, Assisted by Bayes} approach is a principled method for obtaining efficient confidence regions with frequentist guarantees in the presence of prior information.
This article developed a number of properties of this approach.
First, we established simple, heavy-tailed prior conditions that guarantee \emph{uniformly bounded} FAB-CRs -- eliminating the pathological blow-up that can plague the classical Gaussian setup and ensuring the usefulness of the approach even under severe prior–data conflict.
Second, we introduced the notion of a \emph{focal point} for nested FAB procedures and showed that, in Gaussian and, more generally, exponential-family models, the focal point coincides with a transformation of a posterior mean shrinkage estimator.
This puts point estimation and region estimation on the same footing and offers practitioners an immediately interpretable summary to accompany every confidence region.
Third, we proposed a shrinkage-prior construction, implemented through scale mixtures of normals, that simultaneously
\begin{enuminline}
    \item encourages sparsity,
    \item leads to uniformly bounded FAB-CRs, and
    \item reverts to classical intervals in the extremal, data-dominated regime.
\end{enuminline}
\looseness=-1

Pratt's approach minimizes the Bayes expected volume \eqref{eq:Bayesexpectedvol}, under some prior $\pi$, subject to the pointwise frequentist coverage constraint \eqref{eq:coverage}, which is a strong guarantee. In some applications, weaker or different guarantees may be preferred. In the presence of prior information $\pi$, two common alternatives are Bayesian HPD regions and marginal coverage intervals. The former minimize, for each fixed data $y$, the volume $\vol(C_{\alpha}(y))$ under the posterior constraint $\Pr(\theta\in C_{\alpha}(y)\mid y)\geq 1-\alpha$. Alternatively, the latter minimize the Bayes expected volume \eqref{eq:Bayesexpectedvol} under the (weaker) marginal coverage constraint \citep{Joseph1995}
$$
    \Pr(\theta\in C_\alpha(Y))=\int_\Theta\Pr(\theta\in C_\alpha(Y)\mid \theta)\pi(\theta)d\theta \geq 1-\alpha.
$$
In this case, the regions take the form $C_\alpha(y)=\left\{\theta \mid \pi(\theta\mid y)\geq c    \right \}$ for some constant cut-off $c$.
That typically means that the posterior mass $\Pr(\theta\in C_{\alpha}(y)\mid y)$ now varies with $y$.

All numerical experiments concerning the Gaussian likelihood in this work assume \emph{known} variance.
In practice, the variance is usually unknown, and more conservative procedures may therefore be preferable.
One promising direction is to adapt the FAB framework for Student-$t$ confidence intervals, as proposed by \citet{Yu2018}.
Moreover, our results in \cref{sec:gaussianlikelihood,sec:shrinkagepriors} are presented for univariate Gaussian likelihoods.
Extending them to the multivariate case is straightforward under product priors when one is interested in marginal FAB-CRs for each dimension of $\theta \in \mathbb{R}^d$.
In this case, the resulting FAB-CRs are obtained by applying the construction described in \cref{sec:fabbackground} separately to each component of $\theta$, and our results apply dimension-wise.
Indeed, our experimental results in \cref{sec:linearregression} already illustrate this setup.
Going forward, extensions to joint FAB-CRs for multivariate parameters under non-product priors represent a valuable yet more challenging research direction. 
Lastly, our robustness results were derived only for Gaussian models; extending them to other likelihood families remains an open and worthwhile avenue for future work.

% Acknowledgements
\section*{Acknowledgements}
The authors would like to thank Gunnar Taraldsen for valuable discussion and feedback on an earlier draft of this work.
Stefano Cortinovis is supported by the EPSRC Centre for Doctoral Training in Modern Statistics and Statistical Machine Learning (EP/S023151/1).

\bibliographystyle{plainnat}
\bibliography{references}

\newpage
\appendix

\makeatletter
\renewcommand \thesection{S\@arabic\c@section}
\renewcommand\thetable{S\@arabic\c@table}
\renewcommand \thefigure{S\@arabic\c@figure}
\renewcommand{\theequation}{S\arabic{equation}}
\makeatother

\section{Derivation of \texorpdfstring{\Cref{eq:confidenceregionCgeneral}}{Equation (3)}}
\label{sec:derivationCgeneral}

For completeness, we provide a self-contained derivation of the general form of the FAB confidence region $C_\alpha(y)$ given in \cref{eq:confidenceregionCgeneral}.
As in \cref{sec:fabbackground}, let $\calY$ have density $f_\theta$ indexed by a parameter of interest $\theta \in \Theta \subseteq \mathbb{R}^d$.
For a fixed $\alpha \in (0, 1)$ and a prior $\pi$ on $\Theta$, the FAB confidence region $C_\alpha : \calY \mapsto C_\alpha(y) \subseteq \Theta$ is defined as the solution to the constrained optimization problem
\begin{equation}
    \underset{C_\alpha}{\argmin}\ R(C_\alpha) := \int_\Theta \E\left[\vol(C_\alpha(Y)) \mid \theta\right] \pi(d\theta) \text{ s.t. } \Pr(\theta \in C_\alpha(Y) \mid \theta = \theta_0) = 1 - \alpha
    \label{eq:fab_opt}
\end{equation}
for all $\theta_0 \in \Theta$, where $\vol(\cdot)$ denotes the Lebesgue measure on $\Theta$.

By applying Fubini-Tonelli's theorem, $R(C_\alpha)$ can be rewritten as
\begin{align*}
    R(C_\alpha) &= \int_\Theta \int_\calY \vol(C_\alpha(y)) f_\theta(y) dy \pi(d\theta) \\
    &= \int_\calY \vol(C_\alpha(y)) \left( \int_\Theta f_\theta(y) \pi(d\theta) \right) dy \\
    &= \int_\calY \vol(C_\alpha(y)) f(y) dy,
\end{align*}
where $f(y) = \int_\Theta f_\theta(y) \pi(d\theta)$ is the marginal density of $Y$ under $\pi$.
By expanding $\vol(C_\alpha(y))$ and applying Fubini-Tonelli's theorem again, we have that
\begin{equation*}
    R(C_\alpha) = \int_\Theta \int_\calY \mathbf{1}\{\theta_0 \in C_\alpha(y)\} f(y) dy d\theta_0.
\end{equation*}
Define the acceptance region $A_\alpha(\theta_0)=\{y \mid \theta_0\in C_\alpha(y)\}$.
It follows that $\mathbf 1\{\theta_0\in C_\alpha(y)\} = \mathbf 1\{y\in A_\alpha(\theta_0)\}$, and thus
\begin{equation}
    R(C_\alpha)=\int_\Theta \Pr\left(Y\in A_\alpha(\theta_0)\right) d\theta_0, \label{eq:RCAcceptanceRegion}
\end{equation}
where the probability inside the integral is computed with respect to the marginal density $f(y)$.
The equivalence between the expressions for $R(C_\alpha)$ in \cref{eq:fab_opt} and \cref{eq:RCAcceptanceRegion} is referred to as the Ghosh--Pratt identity \citep{Ghosh1961,Pratt1961}, and is a special case of Robbins' theorem \citep{Robbins1944,Robbins1945} on the expected measure of random sets \citep[Theorem 1.5.16]{Molchanov2017}.

With this, the optimization problem \eqref{eq:fab_opt} is equivalent to
\begin{equation*}
    \underset{A_\alpha}{\argmax} \int_\Theta \Pr\left(Y \notin A_\alpha(\theta_0)\right) d\theta_0 \text{ s.t. } \Pr(Y \in A_\alpha(\theta_0) \mid \theta = \theta_0) = 1 - \alpha
\end{equation*}
for all $\theta_0 \in \Theta$.
That is, for each $\theta_0 \in \Theta$, $A_\alpha(\theta_0)$ is the acceptance region of a Neyman--Pearson most powerful test for $H_0: \theta = \theta_0$ against $H_1: \theta \sim \pi$.
By the Neyman--Pearson lemma, $A_\alpha(\theta_0)$ can be obtained as a likelihood-ratio threshold, i.e.,
\begin{equation*}
    A_\alpha(\theta_0) = \left\{y \mid \log \frac{f(y)}{f_{\theta_0}(y)} \le k_\alpha(\theta_0) \right\},
\end{equation*}
where $k_{\alpha}(\theta_0)$ is the smallest value such that $\Pr(Y\in \AR(\theta)\mid \theta=\theta_0)=1-\alpha$.
Finally, the FAB confidence region $C_\alpha(y)$ is obtained by inverting the definition of $A_\alpha(\theta_0)$, i.e.,
\begin{equation*}
    C_\alpha(y) = \{\theta_0 \in \Theta \mid y \in A_\alpha(\theta_0)\}.
\end{equation*}

\section{Proofs of \texorpdfstring{\cref{sec:gaussianlikelihood}}{Section 3}}
\subsection{Proof of \texorpdfstring{\cref{thm:postmeaninCI}}{Theorem 3.1}}

For any $\theta_0\in\bbR$, the function
\begin{align*}
    \lambda_{\theta_0}(y) = \lambda_0(y) -\frac{y\theta_0}{\sigma^2} + \frac{{\theta_0}^2}{2\sigma^2}
\end{align*}
is differentiable, with derivative
\begin{align*}
    \lambda'_{\theta_0}(y) = \lambda'_0(y) - \frac{\theta_0}{\sigma^2}.
\end{align*}
From Tweedie's formula \citep{Efron2011}, we have that
\begin{align*}
    \lambda'_0(y) &= \ell'(y) + \frac{y}{\sigma^2} = \frac{1}{\sigma^2}\E[\theta\mid y] = \frac{1}{\sigma^2}\widehat\theta(y), \\
    \lambda''_0(y) &= \ell''(y) + \frac{1}{\sigma^2} = \frac{1}{\sigma^4}\var(\theta\mid y).
\end{align*}
Moreover, since the prior $\pi$ is non-degenerate, $\var(\theta\mid y) > 0$.
As a result, $\lambda_0$ (as well as $\lambda_{\theta_0}$, for any $\theta_0\in\bbR$) is strictly convex, and $\lambda'_0(y)$ and $\widehat\theta(y)$ are both strictly increasing and bijective functions.

Solving $\lambda'_{\theta_0}(y_0) = 0$ yields $y_0 = {\widehat\theta}^{-1}(\theta_0)\in \AR(\theta_0)$.
Therefore, by inversion, $\widehat{\theta}(y)\in \CR(y)$ for all $\alpha\in(0,1)$.
Furthermore, in this case, $\vol(C_\alpha(y))\to 0$ as $\alpha\to 1$.
Hence, the focal point $\widehat\theta^\FAB(y)$ is defined as in \cref{eq:focalpointFABCR}, and $\widehat\theta^\FAB(y) = \widehat\theta(y)$, as desired.

Additionally, the strict convexity of $\lambda_{\theta_0}(y)$ implies that $\AR(\theta_0)$ is an interval, written $[\loA(\theta_0), \upA(\theta_0)]$, with $\loA(\theta_0)<{\widehat\theta}^{-1}(\theta_0)<\upA(\theta_0)$. Due to the Neyman--Pearson construction \eqref{eq:acceptanceregion_neymanpearson}, the functions $\loA(\theta_0)$ and $\upA(\theta_0)$ satisfy the identities
\begin{align}
    \lambda_{\theta_0}(\upA(\theta_0))&=\lambda_{\theta_0}(\loA(\theta_0)), \label{eq:ulcond1}\\
    \Phi\left(\frac{\upA(\theta_0)-\theta_0}{\sigma}\right)-\Phi\left(\frac{\loA(\theta_0)-\theta_0}{\sigma}\right)&=1-\alpha.\label{eq:ulcond2}
\end{align}
Defining $w_\alpha(\theta_0)$ as in \cref{eq:wfunction} and solving it for $\loA(\theta_0)$, we can write $\loA(\theta_0)$ as
\begin{align*}
    \loA(\theta_0) = \theta_0 - \sigma \Phi^{-1}(1 - \alpha w_\alpha(\theta_0)).
\end{align*}
Similarly, by plugging this expression for $\loA(\theta_0)$ into \cref{eq:ulcond2} and solving it for $\upA(\theta_0)$, we can write $\upA(\theta_0)$ as
\begin{align}
    \upA(\theta_0) &= \theta_0 + \sigma \Phi^{-1}\left(\Phi\left(\frac{\loA(\theta_0) - \theta_0}{\sigma}\right) + (1 - \alpha)\right) \label{eq:upA_loA}\\
    &=\theta_0 + \sigma \Phi^{-1}(1 - \alpha (1 - w_\alpha(\theta_0))). \nonumber
\end{align}
Then, by inverting the acceptance intervals $\AR(\theta_0) = [\loA(\theta_0), \upA(\theta_0)]$ through \cref{eq:confidenceregionCgeneral}, we obtain the expression for the confidence region $\CR(y)$ in \cref{eq:confidenceregionCw}.

Finally, as shown in \cref{eq:upA_loA}, $\upA(\theta_0)$ can be written as $\upA(\theta_0)=h_\theta(\loA(\theta_0))$, where
\begin{align*}
    h_{\theta_0}(l)=\theta_0 + \sigma\Phi^{-1}\left(\Phi\left(\frac{l-\theta_0}{\sigma}\right) + (1-\alpha)\right)
\end{align*}
That is, both $\upA(\theta_0)$ \eqref{eq:upA_loA} and $w_\alpha(\theta_0)$ \eqref{eq:wfunction} may be expressed as the composition of $\loA(\theta_0)$ by a continuous and differentiable function.
As a result of this, it is enough to show the continuity and differentiability of $\loA(\theta_0)$ to conclude that the same property also holds for $\upA(\theta_0)$ and $w_\alpha(\theta_0)$.
To prove the former, note that, by \cref{eq:ulcond1}, $\loA(\theta_0)$ is the solution to the implicit equation
$$
    G(\theta_0,\loA(\theta_0))=0,
$$
where
$$
    G(\theta_0,l)=\lambda_{\theta_0}(l)-\lambda_{\theta_0}\left(h_{\theta_0}(l)\right).
$$
Moreover, the function $h_{\theta_0}(l)$ is differentiable, with $h'_{\theta_0}(l)>0$.
Then, we have that
\begin{align*}
    \frac{\partial G(\theta_0, l)}{dl}&=\lambda'_{\theta_0}(l)-\lambda'_{\theta_0}\left(h_{\theta_0}(l)\right)\times h'_{\theta_0}(l).
\end{align*}
Additionally, $\lambda'_{\theta_0}(\loA(\theta_0))<0$ and $\lambda'_{\theta_0}(\upA(\theta_0))>0$.
Therefore,
$$
    \left.\frac{\partial G(\theta_0,l)}{dl}\right|_{(\theta_0,\loA(\theta_0))}<0.
$$

As a result of this, by the implicit function theorem, $\loA(\theta_0)$ is continuously differentiable, as desired.

\subsection{Proof of \texorpdfstring{\cref{thm:robustCR}}{Theorem 3.2}}
\subsubsection{Asymptotic behavior of \texorpdfstring{$\widehat\theta^\FAB(y)$}{theta-hat-FAB(y)}}
\label{sec:asympthetafab}

From Tweedie's formula,
\begin{equation*}
    \widehat\theta^\FAB(y)-y=\sigma^2 \ell'(y).
\end{equation*}
Let
$$
    h(y)=f(y)e^{\frac{\kappa}{\sigma}|y|}.
$$
The function $h$ is differentiable on $\bbR\backslash \{0\}$, with $h'$ ultimately monotone, and $h(y)\sim \gamma |y|^{-\delta}$.
From \citet[Theorem 1.7.2, p.39]{Bingham1989}
\begin{equation*}
    h'(y)\sim -\delta\gamma y^{-\delta-1}
\end{equation*}
as $y\to \infty$. Hence,
$$
    (\log(h(y)))'=\frac{h'(y)}{h(y)}\sim -\delta|y|^{-1}\to 0
$$
as $y\to \infty$. Similarly, $(\log(h(y)))' \to 0$ as $y\to-\infty$. It follows that
$$
    \ell'(y)\to \left \{\begin{tabular}{ll}
        $-\frac{\kappa}{\sigma}$  & \text{as }$y\to\infty $ \\
        $\frac{\kappa}{\sigma}$ & \text{as }$y\to-\infty $
    \end{tabular}\right.
$$
and
$$
    \widehat\theta^\FAB(y)-y\to \left \{\begin{tabular}{ll}
        $-\kappa\sigma$  & \text{as }$y\to\infty $ \\
        $\kappa\sigma$ & \text{as }$y\to-\infty $ .
    \end{tabular}\right.
$$

\subsubsection{Asymptotic behavior of \texorpdfstring{$C_\alpha(y)$}{C-alpha(y)}}

By \cref{thm:postmeaninCI}, $\AR(\theta_0)=[\loA(\theta_0),\upA(\theta_0)]$, where  $\loA(\theta_0)$ and $\upA(\theta_0)$ satisfy, for any fixed $\theta_0$, \cref{eq:ulcond1,eq:ulcond2}.
As per \cref{eq:confidenceregionCw}, the FAB confidence region $\CR(y)$ is given by
\begin{align*}
    \CR(y) = \{\theta_0\mid \loA(\theta_0)=\theta_0-\sigma \Phi^{-1}(1-\alpha w_\alpha(\theta_0)) \leq y \leq \theta_0 + \sigma \Phi^{-1}(1-\alpha(1-w_\alpha(\theta_0))) = \upA(\theta_0) \},
\end{align*}
where $w_\alpha$ is defined as in \cref{eq:wfunction}.
The following lemma, proved in \cref{sec:prooflemmawbounded}, describes the behavior of $w_\alpha(\theta_0)$ under the assumptions of \cref{thm:robustCR} as $|\theta_0|\to\infty$.
\begin{lemma}\label{thm:wbounded}
    Let $w_\alpha$ be defined as in \cref{eq:wfunction}.
    Under the assumptions of \cref{thm:robustCR}, $w_\alpha(\theta_0)$ is bounded away from 0 and 1, and
    \begin{align}
        \lim_{\theta_0\to -\infty} w_\alpha(\theta_0)=1-\lim_{\theta_0\to \infty} w_\alpha(\theta_0)= c_\alpha,\label{eq:limitcalpha}
    \end{align}
    where $c_\alpha=g_\alpha^{-1}(-2\kappa)\in(0,1/2]$ with $g_\alpha$ defined in \cref{eq:galpha}.
\end{lemma}
By \cref{thm:wbounded}, $w_\alpha$ is bounded away from 0 and 1, and therefore
$$
    \sup_{\theta_0\in\bbR}|\Phi^{-1}(1-\alpha w_\alpha(\theta_0))|\leq M,~~~\sup_{\theta_0\in\bbR}|\Phi^{-1}(1-\alpha(1-w_\alpha(\theta_0)))| \leq M
$$
for some $M>0$. Hence, by \cref{prop:Hausdorffconvergence}, \cref{eq:limitcalpha} and the continuity of $\Phi^{-1}$,
$$
    \lim_{y\to\infty}  \CR(y)-y
    =
    \left[-\sigma \Phi^{-1}(1-\alpha c_\alpha),\ \sigma \Phi^{-1}(1-\alpha (1-c_\alpha)) \right]
$$
where the convergence is with respect to the Hausdorff distance. Similarly,
\[
    \lim_{y\to-\infty}  \CR(y)-y
    =
    \left[-\sigma \Phi^{-1}(1-\alpha (1-c_\alpha)),\ \sigma \Phi^{-1}(1-\alpha c_\alpha) \right].
\]

\subsubsection{Proof of \texorpdfstring{\cref{thm:wbounded}}{Lemma S2.1}}
\label{sec:prooflemmawbounded}

The proof is split in two parts. We first prove that $w_\alpha$ is bounded away from 0 and 1. Then we prove  that
\begin{equation*}
    \lim_{\theta\to -\infty} w_\alpha(\theta)=1-\lim_{\theta\to \infty} w_\alpha(\theta)= c_\alpha
\end{equation*}
where $c_\alpha=g_\alpha^{-1}(-2\kappa)$. We first introduce some notations.

Let $\alpha\in(0,1).$ For any $\omega\in(0,1)$, let
\begin{align*}
    \mu_1(\omega)&=z_{1-\alpha(1-\omega)},\\
    \mu_2(\omega)&=z_{1-\alpha \omega},
\end{align*}
which take values in $[z_{1-\alpha},\infty)$. Recall from \cref{eq:confidenceregionCw} that $\upA(\theta)=\theta + \sigma\mu_1(w_\alpha(\theta))$ and $\loA(\theta)=\theta-\sigma\mu_2(w_\alpha(\theta))$, where $\loA(\theta)$ and $\upA(\theta)$ are respectively the lower bound and upper bound of the acceptance interval $A_\alpha(\theta)$; see \cref{fig:CI_proof} for an illustration.

\begin{figure}
  \centering
  \includegraphics[width=0.75\textwidth]{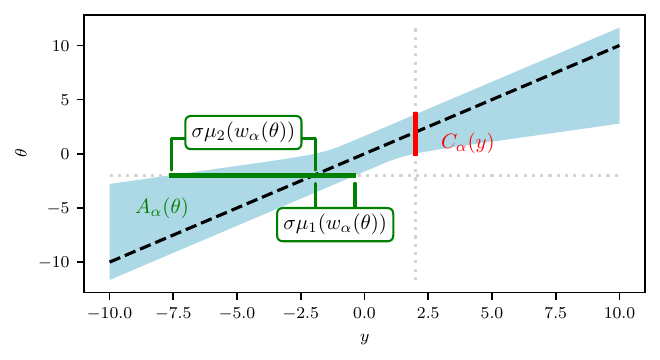}
  \caption{Illustration of the notations used in proof of \cref{thm:wbounded}. For the acceptance interval $A_\alpha(\theta)=[\loA(\theta),\upA(\theta)]$, we have $\upA(\theta)=\theta+\sigma\mu_1(w_\alpha(\theta))$ and $\loA(\theta)=\theta-\sigma\mu_2(w_\alpha(\theta))$.}
  \label{fig:CI_proof}
\end{figure}

\paragraph{Proof that $w_\alpha$ is bounded away from 0 and 1}

We prove that $w_\alpha(\theta)$ is bounded away from $0$. The proof for $1$ is analogous.

Assume for contradiction that $w_\alpha$ is not bounded away from $0$. Then there exists a sequence $\theta_n$ such that
\[
    \omega_n:=w_\alpha(\theta_n)\to0.
\]
Set
\[
    a_n=\theta_n+\sigma\mu_1(\omega_n),\qquad
    b_n=\theta_n-\sigma\mu_2(\omega_n).
\]
Since $\omega_n\to0$, we have $\mu_1(\omega_n)\to z_{1-\alpha}$ and $\mu_2(\omega_n)\to\infty$.

By the defining equation for $w_\alpha$,
\[
    f(a_n)\exp\!\left(\frac{\mu_1(\omega_n)^2}{2}\right)
    =
    f(b_n)\exp\!\left(\frac{\mu_2(\omega_n)^2}{2}\right),
\]
hence
\[
    \frac{f(a_n)}{f(b_n)}
    =
    \exp\!\left(\frac{\mu_2(\omega_n)^2-\mu_1(\omega_n)^2}{2}\right).
\]

Because $f$ is bounded on $\mathbb R$, say by $M$,
\[
    f(b_n)\le
    M\exp\!\left(\frac{\mu_1(\omega_n)^2-\mu_2(\omega_n)^2}{2}\right)\to0.
\]
Since $f$ is continuous and strictly positive on $\mathbb R$, this implies $|b_n|\to\infty$.

If $\delta=\kappa=0$, then \eqref{eq:regvar1} gives $f(y)\to\gamma>0$ as $|y|\to\infty$, contradicting $f(b_n)\to0$. Therefore, from now on, $\delta>0$ or $\kappa>0$.

Using \eqref{eq:regvar1} together with the bound above,
\[
    |b_n|^{-\delta}\exp\!\left(-\frac{\kappa}{\sigma}|b_n|\right)
    = O\!\left(\exp\!\left(-\frac{\mu_2(\omega_n)^2}{2}\right)\right).
\]
Taking logarithms shows that
\[
    \mu_2(\omega_n)=o(|b_n|).
\]
Since
\[
    a_n-b_n=\sigma\bigl(\mu_1(\omega_n)+\mu_2(\omega_n)\bigr)=O(\mu_2(\omega_n)),
\]
it follows that
\[
    a_n-b_n=o(|b_n|).
\]
Hence, $a_n$ and $b_n$ have the same sign for $n$ large enough, and
\[
    \frac{|a_n|}{|b_n|}\to1.
\]
Therefore,
\[
    \frac{f(a_n)}{f(b_n)}
    \sim
    \left(\frac{|a_n|}{|b_n|}\right)^{-\delta}
    \exp\!\left(-\frac{\kappa}{\sigma}\bigl(|a_n|-|b_n|\bigr)\right)
    =
    \exp\bigl(O(\mu_2(\omega_n))\bigr)
    =
    o\!\left(\exp\!\left(\frac{\mu_2(\omega_n)^2}{2}\right)\right).
\]
But the boundary equation gives
\[
    \frac{f(a_n)}{f(b_n)}
    =
    \exp\!\left(\frac{\mu_2(\omega_n)^2-\mu_1(\omega_n)^2}{2}\right)
    \asymp
    \exp\!\left(\frac{\mu_2(\omega_n)^2}{2}\right),
\]
because $\mu_1(\omega_n)$ is bounded. This is a contradiction.

\paragraph{Asymptotic behavior of $w_\alpha(\theta)$}

We now proceed to the proof that $\lim_{\theta\to -\infty} w_\alpha(\theta)=c_\alpha$, where $c_\alpha=g_\alpha^{-1}(-2\kappa)$.

As $w_\alpha$ is bounded away from 0 and 1, the functions $\mu_1(\omega)$ and $\mu_2(\omega)$ are bounded. Hence, $|\theta +\sigma\mu_1(w_\alpha(\theta))|\sim |\theta -\sigma\mu_2(w_\alpha(\theta))|\sim |\theta |$ as $|\theta|\to\infty$. It follows from \eqref{eq:regvar1} that, as $\theta\to -\infty$,
\[
    \frac{f\left(  \theta+\sigma\mu_1(w_\alpha(\theta)) \right)  }{f\left(  \theta-\sigma \mu_2(w_\alpha(\theta))\right)  }\sim e^{\kappa(\mu_1(w_\alpha(\theta))+\mu_2(w_\alpha(\theta)))}.
\]
It follows that, for any $\epsilon_0>0$, there exists $\theta_0\in\bbR$ such that for all $\theta<\theta_0$, $\omega=w_\alpha(\theta)$ satisfies
\[
    1-\epsilon_0 \leq\exp\left(\frac{1}{2}\mu_2(\omega)^2-\frac{1}{2}\mu_1(\omega)^2-\kappa(\mu_1(\omega)+\mu_2(\omega))   \right)\leq1+\epsilon_0.
\]

Consider the continuous and differentiable function $h(\omega)=e^{u(\omega)}$ where
$$
    u(\omega)=\frac{1}{2}\mu_2(\omega)^2-\frac{1}{2}\mu_1(\omega)^2-\kappa(\mu_1(\omega)+\mu_2(\omega)).
$$
Then
\begin{align*}
    u'(\omega)=-\alpha\left(\frac{\mu_1(\omega)+\kappa}{\phi(\mu_1(\omega))} + \frac{\mu_2(\omega)-\kappa}{\phi(\mu_2(\omega))} \right).
\end{align*}
We now use the following lemma, proved at the end of this subsection.
\begin{lemma}\label{th:lemmamu1mu2bound}
    Let $\kappa\geq 0$ be the constant appearing in \eqref{eq:regvar1}.
    Under the assumptions of \cref{thm:robustCR}, for every $\alpha\in(0,1)$,
    there exists $\theta_1\in\bbR$ such that, for all $\theta<\theta_1$,
    \begin{align*}
        \mu_1(w_\alpha(\theta))+\kappa&>0,\\
        \mu_2(w_\alpha(\theta))-\kappa&>0.
    \end{align*}
\end{lemma}
By \cref{th:lemmamu1mu2bound}, taking $\theta_0$ small enough, we have $u'(\omega)<0$ for any $\omega=w_\alpha(\theta)$ with $\theta<\theta_0$. Hence, $h$ is continuous, strictly monotone on $w_\alpha((-\infty,\theta_0))$, therefore has a continuous inverse $h^{-1}$. We are interested in $h^{-1}(1)$, which is the solution to $u(\omega)=0$, or $\mu_2(\omega)-\mu_1(\omega)=2\kappa$. Noting that $\mu_1(\omega)-\mu_2(\omega)=g_\alpha(\omega)$, where $g_\alpha$ is defined in \cref{eq:galpha}, we obtain $h^{-1}(1)=g_\alpha^{-1}(-2\kappa)=c_\alpha$. By continuity of $h^{-1}$, for any $\varepsilon>0$, there exists $\delta > 0$ such that $\left\vert h(\omega)-1\right\vert < \delta$ implies $\left\vert \omega-h^{-1}(1)\right\vert =\left\vert \omega-c_\alpha\right\vert <\varepsilon$. This completes the proof.

\begin{proof}[Proof of \cref{th:lemmamu1mu2bound}]
    By the first part of the proof of \cref{thm:wbounded}, $w_\alpha(\theta)$ is
    bounded away from $0$ and $1$ as $\theta\to-\infty$. Hence, there exist
    $\epsilon\in(0,1/2)$ and $\theta_0\in\bbR$ such that
    \[
        w_\alpha(\theta)\in[\epsilon,1-\epsilon]
        \qquad\text{for all }\theta<\theta_0.
    \]
    Let
    \[
        S(\omega)=\mu_1(\omega)+\mu_2(\omega).
    \]
    For every $\omega\in(0,1)$, $S(\omega)>0$. Indeed, writing $p=\alpha(1-\omega)$ and $q=\alpha\omega$, we have $p+q=\alpha<1$, and hence
    \[
        \Phi^{-1}(1-p)>-\Phi^{-1}(1-q).
    \]
    Thus $S(\omega)>0$. By continuity and compactness,
    \[
        m:=\inf_{\omega\in[\epsilon,1-\epsilon]} S(\omega)>0.
    \]

    From the boundary condition and the tail equivalence \eqref{eq:regvar1}, as $\theta\to-\infty$ with $\omega=w_\alpha(\theta)$,
    \[
        u(\omega):=
        \frac{1}{2}\mu_2(\omega)^2-\frac{1}{2}\mu_1(\omega)^2
        -\kappa\{\mu_1(\omega)+\mu_2(\omega)\}
        \to 0.
    \]
    Since
    \[
        u(\omega)
        =
        \frac{\mu_1(\omega)+\mu_2(\omega)}{2}
        \{\mu_2(\omega)-\mu_1(\omega)-2\kappa\},
    \]
    it follows that
    \[
        D(\omega):=\mu_2(\omega)-\mu_1(\omega)-2\kappa
        \to 0
        \qquad
        \text{along } \omega=w_\alpha(\theta),\ \theta\to-\infty.
    \]
    Choose $\theta_1\leq \theta_0$ such that
    \[
        |D(w_\alpha(\theta))|<m
        \qquad\text{for all }\theta<\theta_1.
    \]
    Then, for $\omega=w_\alpha(\theta)$,
    \[
        \mu_1(\omega)+\kappa
        =
        \frac{S(\omega)-D(\omega)}{2}
        >
        \frac{m-|D(\omega)|}{2}
        >0,
    \]
    and similarly
    \[
        \mu_2(\omega)-\kappa
        =
        \frac{S(\omega)+D(\omega)}{2}
        >
        \frac{m-|D(\omega)|}{2}
        >0.
    \]
    This proves the claim.
\end{proof}

\subsection{Limiting FAB intervals and the improper exponential-tilt prior}
\label{sec:exp_tilt_limits}

The limits obtained in \cref{thm:robustCR} can be identified with FAB-CRs associated with an improper exponential-tilt prior, as introduced by the next proposition.

\begin{proposition}
    \label{prop:exp_tilt_prior}
    Let $Y\mid \theta\sim \Normal(\theta,\sigma^2)$ and let
    \[
        \pi_\xi(d\theta)=\exp\left(-\frac{\xi\theta}{\sigma}\right)\,d\theta,
        \qquad \xi\in\bbR.
    \]
    Denote by $f_\xi(y)$ the corresponding marginal likelihood and by $\CR^{(\xi)}(y)$ the FAB-CR obtained under
    $\pi_\xi$.
    Then,
    \[
        f_\xi(y)
        =
        \int_\bbR f_\theta(y)\pi_\xi(d\theta)
        =
        \exp\left(-\frac{\xi y}{\sigma}+\frac{\xi^2}{2}\right),
    \]
    and the posterior distribution is
    \[
        \theta\mid y \sim \Normal(y-\sigma\xi,\sigma^2).
    \]
    In particular, the focal point is
    \[
        \widehat\theta_\xi^\FAB(y)=y-\sigma\xi.
    \]

    Furthermore, the spending function is constant in $\theta_0$ and equal to
    \[
        w_{\alpha,\xi}=g_\alpha^{-1}(2\xi).
    \]
    Hence, the acceptance interval is given by
    \[
        \AR^{(\xi)}(\theta_0)
        =
        \left[
        \theta_0-\sigma \Phi^{-1}(1-\alpha w_{\alpha,\xi}),
        \,
        \theta_0+\sigma \Phi^{-1}(1-\alpha(1-w_{\alpha,\xi}))
        \right],
    \]
    and the corresponding FAB-CR takes the form
    \[
        \CR^{(\xi)}(y)
        =
        \left[
        y-\sigma \Phi^{-1}(1-\alpha(1-w_{\alpha,\xi})),
        \,
        y+\sigma \Phi^{-1}(1-\alpha w_{\alpha,\xi})
        \right].
    \]
\end{proposition}

\begin{proof}
    Completing the square yields
    \[
        -\frac{(y-\theta)^2}{2\sigma^2}-\frac{\xi\theta}{\sigma}
        =
        -\frac{(\theta-(y-\sigma\xi))^2}{2\sigma^2}
        -\frac{\xi y}{\sigma}
        +\frac{\xi^2}{2},
    \]
    which proves the expression for $f_\xi(y)$ and shows that
    \[
        \theta\mid y \sim \Normal(y-\sigma\xi,\sigma^2).
    \]
    Therefore, by \cref{thm:postmeaninCI},
    \[
        \widehat\theta_\xi^\FAB(y)=y-\sigma\xi.
    \]

    For fixed $\theta_0$, write
    \[
        \Lambda_{\theta_0,\xi}(y)
        =
        \log f_\xi(y)-\log f_{\theta_0}(y)
        =
        \frac{(y-\theta_0)^2}{2\sigma^2}
        -\frac{\xi}{\sigma}(y-\theta_0)
        +K_{\theta_0,\xi},
    \]
    where $K_{\theta_0,\xi}$ does not depend on $y$. Let
    \[
        \AR^{(\xi)}(\theta_0)
        =
        [\theta_0-\sigma\mu_2,\ \theta_0+\sigma\mu_1].
    \]
    Using the size constraint,
    \[
        \mu_1=\Phi^{-1}(1-\alpha(1-w)),
        \qquad
        \mu_2=\Phi^{-1}(1-\alpha w)
    \]
    for some $w\in(0,1)$. Equality of the log-likelihood ratio at the two endpoints
    gives
    \[
        \Lambda_{\theta_0,\xi}(\theta_0+\sigma\mu_1)
        =
        \Lambda_{\theta_0,\xi}(\theta_0-\sigma\mu_2),
    \]
    that is,
    \[
        \frac{\mu_1^2}{2}-\xi\mu_1
        =
        \frac{\mu_2^2}{2}+\xi\mu_2.
    \]
    Since $\mu_1+\mu_2>0$, this is equivalent to
    \[
        \mu_1-\mu_2=2\xi.
    \]
    Now,
    \[
        \mu_1-\mu_2
        =
        \Phi^{-1}(\alpha w)-\Phi^{-1}(\alpha(1-w))
        =
        g_\alpha(w),
    \]
    so $w=w_{\alpha,\xi}=g_\alpha^{-1}(2\xi)$. The expressions for $\AR^{(\xi)}(\theta_0)$ and $\CR^{(\xi)}(y)$ follow by inversion.
\end{proof}

Because $g_\alpha(1-\omega)=-g_\alpha(\omega)$, we have
\[
    g_\alpha^{-1}(2\kappa)=1-g_\alpha^{-1}(-2\kappa)=1-c_\alpha,
\]
where $c_\alpha=g_\alpha^{-1}(-2\kappa)$ is the constant from
\cref{thm:robustCR}. Therefore, the limiting intervals in
\cref{thm:robustCR} can be rewritten as
\begin{align*}
    \lim_{y\to\infty} \CR(y)-y
    &=
    \CR^{(\kappa)}(0)
    =
    \left[
    -\sigma \Phi^{-1}(1-\alpha c_\alpha),
    \,
    \sigma \Phi^{-1}(1-\alpha(1-c_\alpha))
    \right],\\
    \lim_{y\to-\infty} \CR(y)-y
    &=
    \CR^{(-\kappa)}(0)
    =
    \left[
    -\sigma \Phi^{-1}(1-\alpha(1-c_\alpha)),
    \,
    \sigma \Phi^{-1}(1-\alpha c_\alpha)
    \right].
\end{align*}
In particular, when $\kappa=0$, both limits reduce to the standard
$z$-interval, corresponding to the flat improper prior
$\pi_0(d\theta)=d\theta$. When $\kappa>0$, the limit as $y\to\infty$
is the FAB interval associated with $\pi_\kappa(d\theta)$, whereas the
limit as $y\to-\infty$ is the FAB interval associated with
$\pi_{-\kappa}(d\theta)$.

\section{Proofs of \texorpdfstring{\cref{sec:shrinkagepriors}}{Section 4}}
\subsection{Proof of \texorpdfstring{\cref{prop:w_properties}}{Proposition 4.1}}
\label{app:proofpropwproperties}
    
The scale property in \cref{prop:w_properties} is equivalent to saying that $\pi(d\theta; \sigma)$ is the pushforward of $\pi(d\theta; 1)$ under the map $\theta \mapsto \sigma \theta$. Hence, for any nonnegative measurable function $g: \mathbb{R} \to \mathbb{R}$,
\begin{equation}
    \int_\mathbb{R} g(\theta) \pi(d\theta; \sigma) = \int_\mathbb{R} g(\sigma \theta') \pi(d\theta'; 1)
    \label{eq:pushforward}.
\end{equation}
Then, the marginal likelihood $f(y;\sigma)$ under the prior $\pi(d\theta;\sigma)$ is given by
\begin{align}
    f(y;\sigma)
    &= \int_{\mathbb R} f_\theta(y;\sigma)\,\pi(d\theta;\sigma) \nonumber \\
    &= \int_{\mathbb R} f_{(\sigma \theta')}(y;\sigma)\,\pi(d\theta';1) \nonumber \\
    &= \frac{1}{\sigma}\int_{\mathbb R} f_{\theta'}\left(\frac{y}{\sigma};1\right)\,\pi(d\theta';1) \nonumber \\
    &= \frac{1}{\sigma} f\!\left(\frac{y}{\sigma};1\right),
    \label{eq:ml_locscale}
\end{align}
where the second equality follows from \cref{eq:pushforward}.
By \cref{eq:ulcond1}, $\loA(\theta; \sigma)$ and $\upA(\theta; \sigma)$ under the prior $\pi(d\theta; \sigma)$ satisfy
\begin{equation}
    f(\loA(\theta; \sigma); \sigma) \exp\left(\frac{(\loA(\theta; \sigma) - \theta)^2}{2 \sigma^2}\right) = f(\upA(\theta; \sigma); \sigma) \exp\left(\frac{(\upA(\theta; \sigma)- \theta)^2}{2 \sigma^2}\right)
    \label{eq:ulcond1_locscale}
\end{equation}
under the constraint in \cref{eq:ulcond2}. Furthermore, by \cref{eq:ml_locscale}, \cref{eq:ulcond1_locscale} is equivalent to
\begin{equation}
    f\left(\frac{\loA(\theta; \sigma)}{\sigma}; 1\right) \exp\left(\frac{(\loA(\theta; \sigma)/\sigma - \theta/\sigma)^2}{2}\right) = f\left(\frac{\upA(\theta; \sigma)}{\sigma}; 1\right) \exp\left(\frac{(\upA(\theta; \sigma)/\sigma - \theta/\sigma)^2}{2}\right). \label{eq:ulcond1_locscale2}
\end{equation}
We have that $\loA(\theta; \sigma) = \sigma \loA(\theta/\sigma; 1)$ and $\upA(\theta; \sigma) = \sigma \upA(\theta/\sigma; 1)$.
To see this, plug these expressions into \cref{eq:ulcond1_locscale2}, which becomes
\begin{equation*}
    f\left(\loA(\theta / \sigma; 1); 1\right) \exp\left(\frac{(\loA(\theta/\sigma; 1) - \theta/\sigma)^2}{2}\right) = f\left(\upA(\theta/\sigma; 1); 1\right) \exp\left(\frac{(\upA(\theta/\sigma; 1) - \theta/\sigma)^2}{2}\right).
\end{equation*}
Similarly, \cref{eq:ulcond2} takes the form
\begin{equation*}
    \Phi\left(\upA(\theta/\sigma; 1) - \frac{\theta}{\sigma}\right)-\Phi\left(\loA(\theta/\sigma; 1)- \frac{\theta}{\sigma}\right)=1-\alpha.
\end{equation*}
Since the pair satisfying \cref{eq:ulcond1} and \cref{eq:ulcond2} is unique, the claimed identities for $\loA(\theta;\sigma)$ and $\upA(\theta;\sigma)$ follow.
Finally, by plugging in the expressions for $\loA(\theta; \sigma)$ into \cref{eq:wfunction}, we obtain
\begin{align*}
    w_\alpha(\theta; \sigma) &= \frac{1}{\alpha}\Phi\left(\loA(\theta/\sigma; 1) - \frac{\theta}{\sigma}\right) = w_\alpha(\theta/\sigma; 1),
\end{align*}
as desired. From now on, we suppress the scale parameter $\sigma$ in the notation to indicate $\sigma = 1$.
If $\pi(d\theta)$ is symmetric, i.e., $\pi(A)=\pi(-A)$ for all $A \in \mathcal{B}(\mathbb{R})$, then
\begin{equation*}
    \int_\mathbb{R} g(\theta)\pi(d\theta)=\int_\mathbb{R} g(-\theta)\pi(d\theta)
\end{equation*}
for any nonnegative measurable $g$.
Hence,
\begin{equation*}
    f(-y)=\int_\mathbb{R} f_\theta(-y)\pi(d\theta)
    =\int_\mathbb{R} f_{-\theta}(y)\pi(d\theta)
    =\int_\mathbb{R} f_{\theta}(y)\pi(d\theta) = f(y).
\end{equation*}
This implies that $\loA(-\theta) = -\upA(\theta)$ and $\upA(-\theta) = -\loA(\theta)$. To see this, notice that these expressions for $\loA(-\theta)$ and $\upA(-\theta)$ satisfy \cref{eq:ulcond1} and, moreover, \cref{eq:ulcond2} becomes
\begin{align*}
    \Phi\left(\upA(-\theta) + \theta\right)-\Phi\left(\loA(-\theta)+ \theta\right) &= \Phi\left(-\loA(\theta) + \theta\right) - \Phi\left(-\upA(\theta) + \theta\right) = 1 - \alpha
\end{align*}
because $\Phi(-x) = 1 - \Phi(x)$.
Since the pair satisfying \cref{eq:ulcond1} and \cref{eq:ulcond2} is unique, the identities $\loA(-\theta) = -\upA(\theta)$ and $\upA(-\theta) = -\loA(\theta)$ follow.
By plugging in the expression for $\loA(-\theta)$ into \cref{eq:wfunction} and using the expression for $\upA(\theta)$ in \cref{eq:confidenceregionCw}, this implies that
\begin{equation*}
    w_\alpha(-\theta) = \frac{1}{\alpha}\Phi\left(-\upA(\theta)+\theta\right) = \frac{1}{\alpha}\Phi\left(-\Phi^{-1}\left(1 - \alpha (1 - w_\alpha(\theta))\right)\right) = 1 - w_\alpha(\theta),
\end{equation*}
as desired.

\subsection{Proof of \texorpdfstring{\cref{prop:gcm_power_law}}{Proposition 4.2}}
\label{app:proofgcmpowerlaw}

Let $\Pi_{\tau^2}$ denote the law of $\tau^2$ on $[0,\infty)$, i.e., its Lebesgue--Stieltjes measure $\Pi_{\tau^2}([0,t]) = F_{\tau^2}(t)$.
The marginal likelihood $f(y)$ can be written as
\begin{align*}
    f(y)
    &= \int_{\mathbb{R}} f_\theta(y)\,\pi(d\theta) \\
    &= \int_{\mathbb{R}} \mathcal{N}(y;\theta,\sigma^2)\,
        \left(\int_{[0,\infty)} \mathcal{N}(\theta;0,\sigma^2 t)\,\Pi_{\tau^2}(dt)\right) d\theta \\
    &= \int_{[0,\infty)} \mathcal{N}(y;0,\sigma^2(1+t))\,\Pi_{\tau^2}(dt) \\
    &= \frac{1}{\sqrt{2\pi\sigma^2}}
        \int_{[0,\infty)} (1+t)^{-1/2}\exp\!\left(-\frac{y^2}{2\sigma^2(1+t)}\right)\Pi_{\tau^2}(dt).
\end{align*}

Define $U := 1/(1+\tau^2)\in(0,1]$ and let $\nu$ be the law of $U$ on $(0,1]$
(the pushforward of $\Pi_{\tau^2}$ under $t\mapsto 1/(1+t)$).
Then,
\begin{align}
    f(y)
    &= \frac{1}{\sqrt{2\pi\sigma^2}} \int_{(0,1]} u^{1/2} e^{-u y^2/(2\sigma^2)}\,\nu(du)
    \label{eq:betaprime_marginal_cov} \\
    &= g\left(\frac{y^2}{2\sigma^2}\right), \nonumber
\end{align}
where
\[
    g(x) = \frac{1}{\sqrt{2\pi\sigma^2}} \int_{(0,1]} u^{1/2} e^{-ux}\,\nu(du).
\]

Next note that, for $u\to 0$,
\begin{align*}
    \nu((0,u])
    &= \Pr(U \le u)
        = \Pr\!\left(\tau^2 \ge \frac{1-u}{u}\right)
        = \bar F_{\tau^2}\!\left(\frac{1-u}{u}\right)
    \sim \frac{C_1}{\beta}\,u^{\beta},
\end{align*}
where the last step uses \eqref{eq:tau2_tail_rv} and $(1-u)/u \sim 1/u$.

Define the increasing function
\[
    H(u) := \int_{(0,u]} s^{1/2}\,\nu(ds), \qquad u\in(0,1].
\]
By Stieltjes integration by parts,
\[
    H(u) = u^{1/2}\nu((0,u]) - \frac{1}{2}\int_0^u s^{-1/2}\nu((0,s])\,ds.
\]
Using $\nu((0,s]) \sim (C_1/\beta)s^\beta$ and \cref{prop:rvint} (Karamata's theorem at $0$),
\[
    \int_0^u s^{-1/2}\nu((0,s])\,ds
    \sim \frac{C_1}{\beta}\int_0^u s^{\beta-1/2}\,ds
    = \frac{C_1}{\beta}\cdot \frac{u^{\beta+1/2}}{\beta+1/2},
\]
and therefore
\[
    H(u) \sim \frac{C_1}{\beta+1/2}\,u^{\beta+1/2}
    \qquad \text{as }u\to 0.
\]

Finally, since
\[
    \int_{(0,1]} u^{1/2} e^{-ux}\,\nu(du) = \int_{(0,1]} e^{-ux}\,dH(u),
\]
applying \cref{th:rvlaplace} (Laplace--Stieltjes Tauberian theorem) yields, as $x\to\infty$,
\begin{align*}
    g(x)
    &\sim \frac{1}{\sqrt{2\pi\sigma^2}}
            \Gamma\!\left(\beta+\frac{3}{2}\right)\frac{C_1}{\beta+1/2}\,x^{-(\beta+1/2)} \\
    &= \frac{C_1}{\sqrt{2\pi\sigma^2}} \Gamma\!\left(\beta+\frac{1}{2}\right)\,x^{-(\beta+1/2)},
\end{align*}
where we used $\Gamma(\beta+3/2) = (\beta+1/2)\Gamma(\beta+1/2)$.
Finally, from \cref{prop:rvpower}, we obtain
\[
    f(y) = g\!\left(\frac{y^2}{2\sigma^2}\right)
    \sim C_1 \frac{(2\sigma^2)^\beta}{\sqrt{\pi}} \Gamma\!\left(\beta+\frac{1}{2}\right)\,|y|^{-(2\beta+1)}
\]
as $y\to\pm\infty$, as desired.

\subsection{Proof of \texorpdfstring{\cref{prop:betaprime_marginal}}{Proposition 4.3}}
\label{app:proofbetaprimemarginal}

Starting from the same change of variables as in \cref{eq:betaprime_marginal_cov}, the marginal likelihood $f(y)$ is given by
\begin{align*}
    f(y) &= \frac{1}{\sqrt{2 \pi \sigma^2}} \int_0^1 \exp\left(-u \frac{y^2}{2 \sigma^2}\right) u^{-3/2} f_{\tau^2}\left(\frac{1 - u}{u}\right) du \\
    &= \frac{1}{B(a, b) \sqrt{2 \pi \sigma^2}} \int_0^1 \exp\left(-u \frac{y^2}{2 \sigma^2}\right) u^{-3/2} \left(\frac{1 - u}{u}\right)^{b-1} \left(\frac{1}{u}\right)^{-(a + b)} du \\
    &= \frac{1}{B(a, b) \sqrt{2 \pi \sigma^2}} \int_0^1 \exp\left(-u \frac{y^2}{2 \sigma^2}\right) u^{a - 1/2} (1 - u)^{b-1} du \\
    &= \frac{1}{B(a, b) \sqrt{2 \pi \sigma^2}} \times \frac{\Gamma(b)\Gamma(a + 1 / 2)}{\Gamma(a + b + 1/2)} \times {\OneFOne}\left(a + 1/2, a + b + 1/2, -\frac{y^2}{2 \sigma^2}\right) \\
    &= \frac{1}{\sqrt{2 \pi \sigma^2}} \times \frac{\Gamma(a + b)\Gamma(a + 1 / 2)}{\Gamma(a + b + 1/2)\Gamma(a)} \times {\OneFOne}\left(a + 1/2, a + b + 1/2, -\frac{y^2}{2 \sigma^2}\right)
\end{align*}
because $\OneFOne(\alpha, \beta, z)$ admits the integral representation \citep{Abramowitz1968}
\begin{equation*}
    \OneFOne(\alpha, \beta, z) = \frac{\Gamma(\beta)}{\Gamma(\beta - \alpha) \Gamma(\alpha)} \int_0^1 e^{zu} u^{\alpha - 1} (1 - u)^{\beta - \alpha - 1} du.
\end{equation*}
Moreover, $f'(y)$ is given by
\begin{equation*}
    f'(y) = -\frac{1}{\sqrt{2 \pi \sigma^2}} \times \frac{\Gamma(a + b)\Gamma(a + 3 / 2)}{\Gamma(a + b + 3/2)\Gamma(a)} \times \frac{y}{\sigma^2} \times {\OneFOne}\left(a + 3/2, a + b + 3/2, -\frac{y^2}{2 \sigma^2}\right)
\end{equation*}
because
\begin{equation*}
    \frac{\partial}{\partial z} {\OneFOne}(\alpha, \beta, z) = \frac{\alpha}{\beta} {\OneFOne}(\alpha + 1, \beta + 1, z).
\end{equation*}
As a result of this, by Tweedie's formula \eqref{eq:posteriormean},
\begin{equation*}
    \mathbb{E}[\theta | y] = y + \sigma^2 \frac{f'(y)}{f(y)},
\end{equation*}
from which the expression for $\widehat{\theta}^\text{FAB}(y)$ follows.

\section{Proofs of \texorpdfstring{\cref{sec:exponentialfamily}}{Section 5}}
\subsection{Proof of \texorpdfstring{\cref{thm:postmeaninCI_nef}}{Theorem 5.1}}

Note that $\lambda_\eta(y)$ is a strictly convex function.
Therefore, since $\ARnp$ is a sublevel set of $\lambda_\eta$, it is convex.
Moreover, $\lambda_\eta$ is differentiable w.r.t.~$y$, with  $\nabla\lambda_\eta(y)=\nabla\lambda(y) -\eta.$
Solving $\nabla\lambda_\eta(y)=0$ gives $y=(\nabla\lambda)^{-1}(\eta)$.
Hence, $(\nabla\lambda)^{-1}(\eta)\in \ARnp(\eta)$, which implies that $\nabla\lambda(y)\in\CRnp(y)$.
It follows that $\statfunc(\nabla\lambda(y))\in\CR(y)$.

\section{Auxiliary results}

The following result is stated in \citet[p.333]{Yu2018}.
\begin{proposition}
    For $\alpha\in(0,1]$, let $g_\alpha:(0,1)\to\bbR$ be defined as
    \[
        g_\alpha(\omega)=\Phi^{-1}(\omega \alpha)-\Phi^{-1}(\alpha(1-\omega)).
    \]
    $g_\alpha$ is a continuous and strictly increasing function from $(0,1)$ into $\bbR$, with a well-defined inverse $g_\alpha^{-1}:\bbR\to(0,1)$. In the special case $\alpha=1$, $g_1(\omega)=2\Phi^{-1}(\omega)$ and $g_1^{-1}(\theta)=\Phi(\frac{\theta}{2})$.
\end{proposition}
\begin{proof}
    $\Phi^{-1}$ is continuous and differentiable, hence $g_\alpha$ is continuous and differentiable, with derivative
    $$
        g'_\alpha(\omega)= \frac{\alpha}{\phi(\Phi^{-1}(\omega \alpha))}+\frac{\alpha}{\phi(\Phi^{-1}(\alpha(1-\omega)))}>0
    $$
    where $\phi$ denotes the pdf of the standard normal. By the monotone inverse theorem, it has a well-defined inverse $g_\alpha^{-1}:\bbR\to(0,1)$.
\end{proof}

\begin{proposition}
    \label{prop:Hausdorffconvergence}
    Let $C(y)\subseteq \mathbb R$ be defined as
    \begin{align*}
        C(y)=\left \{\theta = y -\delta  \mid -\ell_1(y-\delta) \leq \delta \leq \ell_2(y-\delta) \right \}
    \end{align*}
    where $\ell_1$ and $\ell_2$ are bounded functions of $\mathbb R$ such that
    \begin{align*}
        \lim_{z\to\infty} \ell_1(z)=c_1,~~~\lim_{z\to\infty} \ell_2(z)=c_2
    \end{align*}
    with $c_2>-c_1$.
    Then
    $$
        \lim_{y\to \infty}C(y)-y=[-c_2,c_1]
    $$
    where the convergence is with respect to the Hausdorff distance on closed subsets of the reals.
\end{proposition}

\begin{proof}
    There is $M>0$ such that $|\ell_1(z)|\leq M$ and $|\ell_2(z)|\leq M$. Hence, $C(y)\subseteq [y-M,y+M]$. For all $\frac{c_1+c_2}{2}>\epsilon>0$, there is $y_0$ such that for all $y>y_0$, $|\delta|\leq M$,
    \begin{align*}
        |\ell_1(y-\delta)-c_1|<\epsilon\text{ and }|\ell_2(y-\delta)-c_2|<\epsilon.
    \end{align*}
    It follows that
    $$
        \{-c_1+\epsilon \leq \delta\leq c_2-\epsilon \}\subseteq\{\delta\mid -\ell_1(y-\delta) \leq \delta \leq \ell_2(y-\delta)\}\subseteq \{-c_1-\epsilon \leq \delta\leq c_2+\epsilon \}
    $$
    and
    $$
        [y-c_2+\epsilon,y+c_1-\epsilon]\subseteq C(y)\subseteq [y-c_2-\epsilon,y+c_1+\epsilon].
    $$
\end{proof}

\section{Background on regular variation}

This section provides background material on regularly varying functions, taken from  \citet{Bingham1989}.

\begin{definition}[Slowly varying function]
    Let $\ell:[0,\infty)\to[0,\infty)$. $\ell$ is said to be a slowly varying function if it satisfies, for any $c>0$
    $$
        \lim_{x\to\infty}\frac{\ell(cx)}{\ell(x)}=1.
    $$
    \label{def:slowvarying}
\end{definition}
Examples of slowly varying functions include logarithms, powers of logarithms, functions converging to a strictly positive constant, etc.

\begin{definition}[Regularly varying function]
    Let $U:[0,\infty)\to[0,\infty)$. $U$ is said to be regularly varying, with index $\rho\in\bbR$ if, for any $c>0$,
    $$
        \lim_{x\to\infty}\frac{U(cx)}{U(x)}=c^\rho.
    $$
    A regularly varying function can always be represented as
    $$
        U(x)=\ell(x) x^\rho
    $$
    for some slowly varying function $\ell$.
    \label{def:regvarying}
\end{definition}

\begin{proposition}
    If $U$ is a regularly varying function with
    $$
        U(x)\sim \ell(x) x^\rho
    $$
    as $x\to\infty$, for some $\rho\in\bbR$ and some slowly varying function $\ell$ then, for any $c>0$ and $a> 0$,
    $$
        U(c x^a)\sim \ell(x^a) c^\rho x^{a\rho}
    $$
    as $x\to\infty$.
    \label{prop:rvpower}
\end{proposition}

The following is a corollary of \citet[Proposition 1.5.10]{Bingham1989}.

\begin{proposition}
    Let $\ell$ be a slowly varying function and $\rho >-1$. Then $\int_0^x t^\rho \ell(1/t)dt$ converges and
    $$
        \int_0^x t^\rho \ell(1/t)dt\sim \frac{x^{\rho+1}}{\rho+1}\ell(1/x)
    $$
    as $x\to 0$.
    \label{prop:rvint}
\end{proposition}

\begin{theorem}[Laplace transforms of regularly varying functions]\cite[Theorem 1.7.1']{Bingham1989} and \cite[Chapter XIII, Section 5]{Feller1971}.
    Let $U$ be non-decreasing on $\bbR$, $U(x)=0$ for all $x<0$ and such that $\widehat U(s)=\int_{-\infty}^\infty e^{-sx} dU(x)<\infty$ for all large $s$. Let $\ell$ be a slowly varying function, $c\geq 0$, $\rho\geq 0$. The following are equivalent
    \begin{align*}
        U(x)&\sim c x^\rho \ell(1/x)/\Gamma(1+\rho)&\text{as }x\to 0\\
        \widehat U(s)&\sim c s^{-\rho} \ell(s)&\text{as }s\to\infty.
    \end{align*}
    \label{th:rvlaplace}
\end{theorem}

\section{Additional figures}\label{app:additional_figures}
\subsection{Gaussian likelihood}
\subsubsection{Priors with vanishing influence}

\Cref{fig:fab_gaussian_vanishing_cc} shows representative $p$-value functions corresponding to the FAB-CR procedures induced by the vanishing influence priors discussed in \cref{sec:vanishing_influence}. 
\begin{figure}
    \centering
    \includegraphics[width=\textwidth]{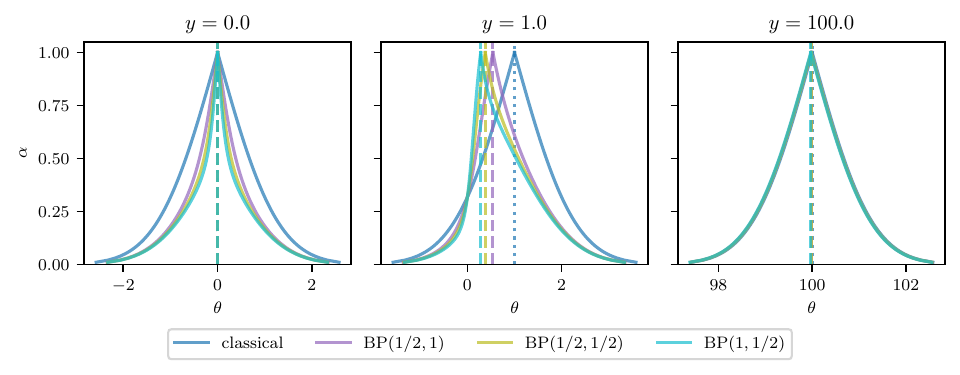}
    \caption{$p$-value functions for the classical $z$-interval and the FAB-CR procedure under the same priors and likelihood of \cref{fig:fab_gaussian_vanishing} when observing $y \in \{0, 1, 100\}$.}
    \label{fig:fab_gaussian_vanishing_cc}
\end{figure}

\subsubsection{Priors with bounded influence}
\label{app:bounded_influence}

\Cref{fig:fab_gaussian_bounded_cc} reports representative FAB-CR $p$-value functions and focal points (i.e., posterior means) for the Laplace priors used in \cref{fig:fab_gaussian_bounded}, as well as the corresponding MAP estimates, indicated with triangular markers.
\begin{figure}
    \centering
    \includegraphics[width=\textwidth]{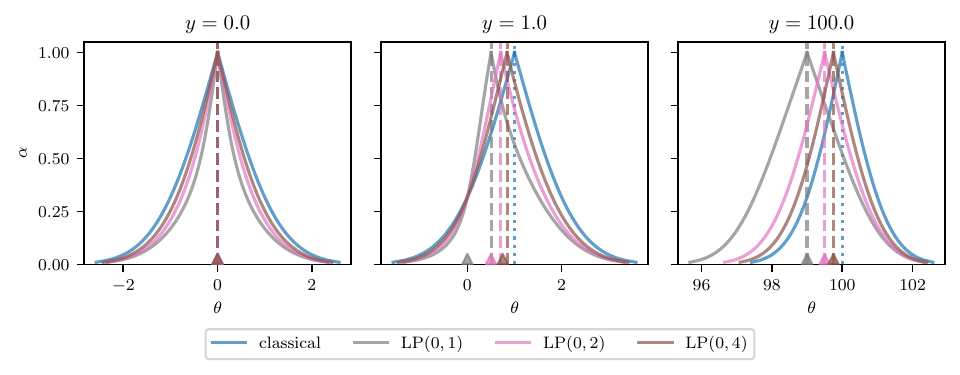}
    \caption{$p$-value functions for the classical $z$-interval and the FAB-CR procedure under the same priors and likelihood of \cref{fig:fab_gaussian_bounded} when observing $y \in \{0, 1, 100\}$.}
    \label{fig:fab_gaussian_bounded_cc}
\end{figure}
As expected, for small non-zero values of $y$ (e.g., $y=1$), FAB and MAP estimates may differ.

\subsection{Natural exponential families}
\subsubsection{Poisson likelihood}

\Cref{fig:fab_poisson_cc} shows representative $p$-value functions corresponding to the FAB-CR procedures in \cref{fig:fab_poisson}.
\begin{figure}
  \centering
  \includegraphics[width=0.8\textwidth]{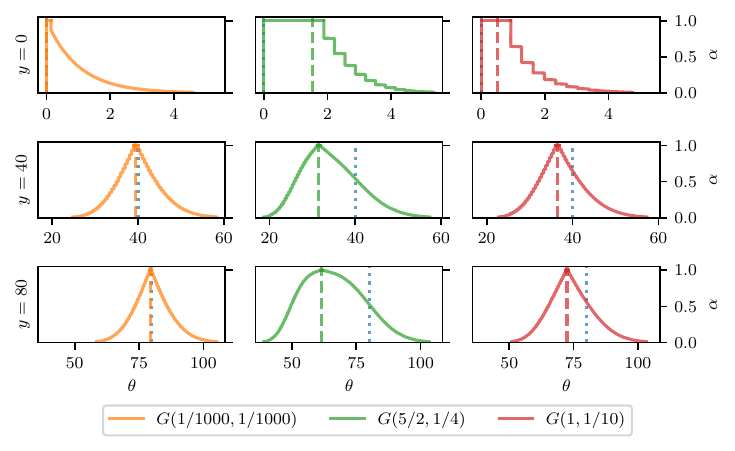}
  \caption{$p$-value functions for the FAB-CR procedure under the same priors and likelihood of \cref{fig:fab_poisson} when observing $y \in \{0, 40, 80\}$.}
  \label{fig:fab_poisson_cc}
\end{figure}
The dotted vertical lines indicate the value of the MLE for $\theta$, $y$, while the dashed vertical lines colored according to the legend indicate the FAB estimates for the same quantity under the corresponding priors.
As prescribed by \cref{thm:postmeaninCI_nef}, the FAB estimate is contained in the confidence region for all values of $\alpha$, while the same is not true for the MLE.

\subsubsection{Binomial likelihood}

Let $Y\sim \Binomial(n, \theta)$ be a Binomial random variable with success probability $\theta$.
It takes the NEF form $f_\eta(y)=h(y) e^{\eta y - \psi(\eta)}$, with $\eta=\log(\theta/(1-\theta))\in\bbR$, $\psi(\eta)=n\log(1+e^\eta)$ and $h(y) = \binom{n}{y}$.
Assuming $\theta\sim \Beta(a,b)$, $\eta = \log(\theta/(1-\theta))$ follows a logistic-beta distribution with density $g(\eta)$, and
\begin{equation*}
    \lambda(y)=\log\int e^{\eta y - \psi(\eta)}g(\eta)d\eta=\log\frac{B(a+y,b+n-y)}{B(a,b)}
\end{equation*}
which is defined for $y\in\widetilde\calY=(-a,b+n)$.
For $y\in\{0,\dots,n\}$, the FAB estimator for $\eta$ is given by
$$
    \widehat\eta^{\FAB}(y) = \E[\eta\mid y] = \lambda'(y) = \digamma(a+y)-\digamma(b+n-y) \in \CRnp(y).
$$
Given that $\theta = \E[Y | \eta] / n = \psi'(\eta)/n = (1 + \exp(-\eta))^{-1}$, the FAB estimator for $\theta$ is given by $\widehat\theta^{\FAB}(y) = (1 + \exp(-\widehat\eta^{\FAB}(y)))^{-1} \in \CR(y)$.
In terms of FAB-CRs, when using a uniform prior ($a=b=1$), one recovers the regions for binomial proportions introduced by \citet{Sterne1954} and further discussed by \citet{Crow1956,Pratt1961}.
\Cref{fig:fab_binomial} illustrates the FAB procedure applied to the binomial likelihood.
\begin{figure}
  \centering
  \includegraphics[width=\textwidth]{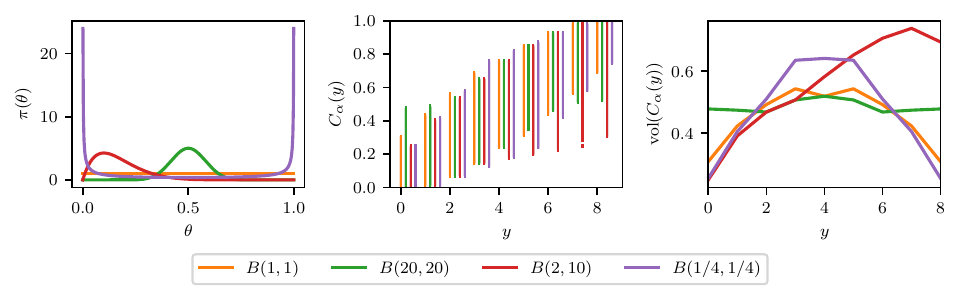}
  \caption{Comparison of the FAB procedure applied to the likelihood $Y \sim \mathrm{Binomial}(n, \theta)$ under a prior $\theta \sim \Beta(a, b)$ for different choices of $a$ and $b$, and where $n = 8$ and $\alpha = 0.1$.}
  \label{fig:fab_binomial}
\end{figure} 
\Cref{fig:fab_binomial_cc} shows representative $p$-value functions corresponding to the FAB-CR procedures in \cref{fig:fab_binomial}.
\begin{figure}
  \centering
  \includegraphics[width=0.8\textwidth]{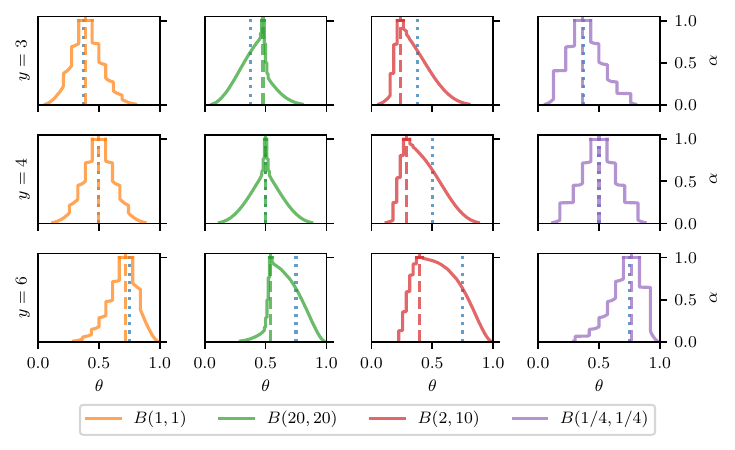}
  \caption{$p$-value functions for the FAB-CR procedure under the same priors and likelihood of \cref{fig:fab_binomial} when observing $y \in \{3, 4, 6\}$.}
  \label{fig:fab_binomial_cc}
\end{figure}
As above, the dotted vertical lines indicate the value of the MLE for $\theta$, $y / n$, while the dashed vertical lines colored according to the legend indicate the FAB estimates for the same quantity under the corresponding priors.
Also here, in accordance with \cref{thm:postmeaninCI_nef} and unlike the MLE, the FAB estimate is contained in the confidence region for all values of $\alpha$.

\subsubsection{Multinomial likelihood}

The construction for the binomial likelihood may be easily generalized to multinomial random variables, where using a uniform prior over the simplex (i.e., Dirichlet prior with parameters equal to 1) recovers the CRs of \citet{Chafai2009}; see also \citet{Malloy2020}.
\Cref{fig:fab_multinomial} illustrates the FAB procedure applied to the multinomial likelihood.
\begin{figure}
  \centering
  \includegraphics[width=0.7\textwidth]{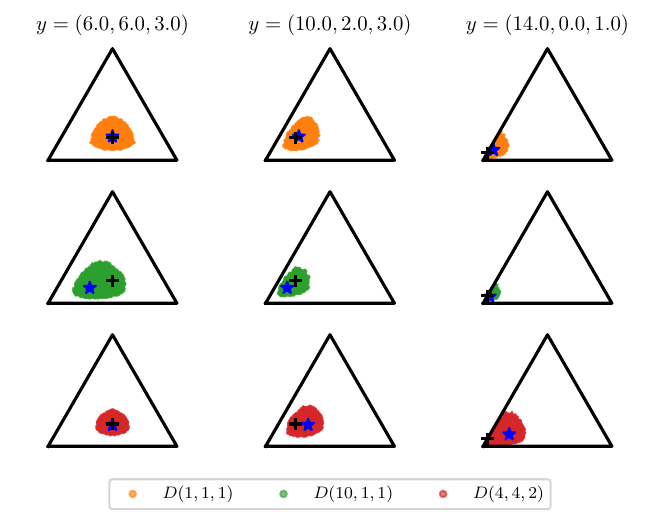}
  \caption{Comparison of the FAB procedure applied to the likelihood $Y \sim \mathrm{Multinomial}(n, k, \boldsymbol{\theta})$ under a prior $\boldsymbol{\theta} \sim \mathrm{Dirichlet}(\mathbf{a})$ for different choices of $\mathbf{a}$, and where $n = 15$, $k =3$ and $\alpha = 0.3$. Each column refers to a different value of the observed $\mathbf{y}$.}
  \label{fig:fab_multinomial}
\end{figure}
Plus (\texttt{+}) and star ($\star$) markers denote the MLE for $\theta$, $\mathbf{y}/n$ and the FAB estimate of $\boldsymbol{\theta}$, respectively.
The difference among the resulting FAB-CRs clearly shows the potential for the procedure to incorporate prior information in the more challenging multivariate setting.

\subsection{Application to linear regression}
\subsubsection{Additional baselines for the linear regression simulation}

\Cref{fig:fab_ols_baselines} adds Bayesian credible regions (highest posterior density regions) under the same priors as in \cref{fig:fab_ols}, as well as desparsified Lasso confidence intervals \citep{Zhang2013}, implemented in the \texttt{hdi} R package \citep{Dezeure2015}.
\begin{figure}
    \centering
    \includegraphics[width=\textwidth]{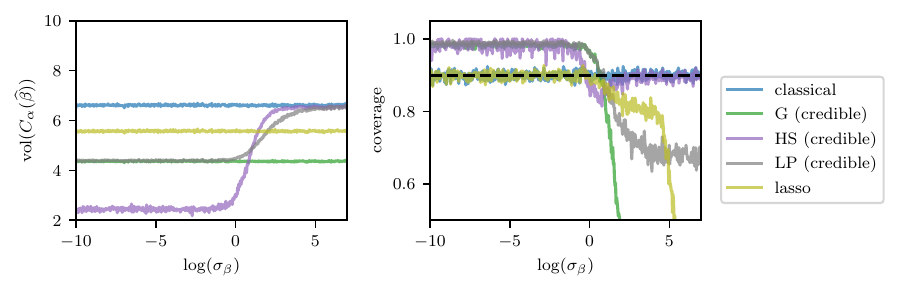}
    \caption{Linear regression simulation study with additional baselines. Comparison of Bayesian and frequentist baselines on the same experiment as in \cref{fig:fab_ols}. Results are averaged over $100$ repetitions.}
    \label{fig:fab_ols_baselines}
\end{figure}
While all methods effectively encourage sparsity, yielding smaller regions than the $z$-interval over a substantial range of $\sigma_\beta^2$ values, none consistently maintains nominal coverage uniformly over $\sigma_\beta^2$ in this experiment, whereas the FAB-CR procedures do so by construction.

\subsubsection{Selected CR realizations for the linear regression real-world study}
\Cref{fig:fab_ols_real_examples} reports selected realizations of the FAB-CRs and $z$-intervals for the linear regression real-world study of \cref{fig:fab_ols_real} in \cref{sec:linearregression}.
\begin{figure}
    \centering
    \includegraphics[width=0.9\textwidth]{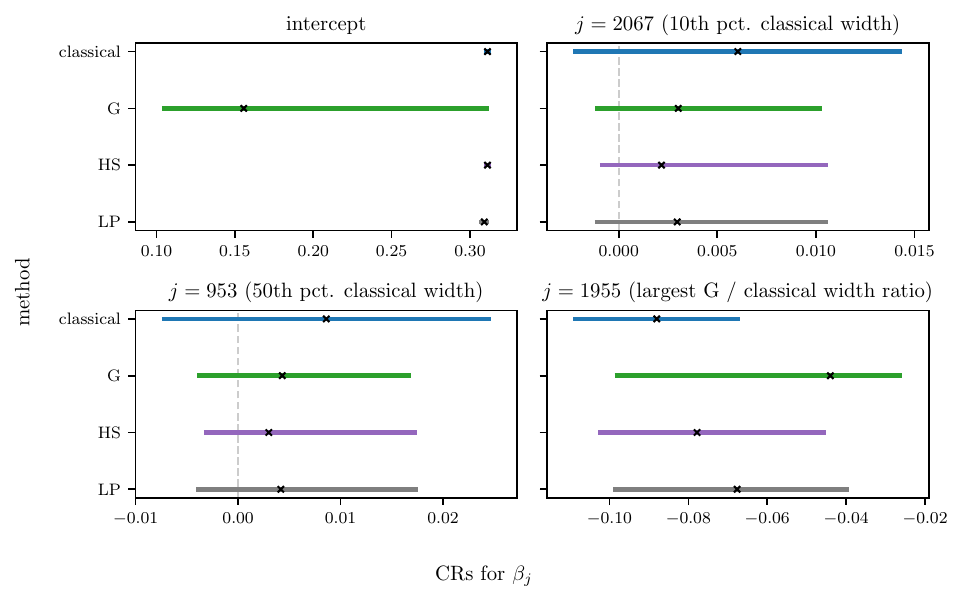}
    \caption{Selected CR realizations for the linear regression real-world study. FAB-CRs and $z$-intervals for the intercept and three representative non-intercept coefficients: the 10th- and 50th-percentile coefficients under the $z$-interval width ranking, and the non-intercept coefficient with the largest discrepancy between Gaussian FAB-CR and $z$-interval widths. Dashed vertical lines indicate $0$, and cross ($\times$) markers indicate each method's focal point, i.e., either the MLE or the FAB estimates.}
    \label{fig:fab_ols_real_examples}
\end{figure}
The plot highlights both the shrinkage induced by FAB near zero and the improved robustness of the heavier-tailed priors under prior--data conflict, as discussed in the main text.

\section{Implementation}
The Python code used to reproduce the figures in this article will be made publicly available upon publication.
All the examples presented here were run locally on an Apple M4 Pro CPU.

\subsection{Special functions and numerical evaluation in \texorpdfstring{\cref{sec:shrinkagepriors}}{Section 4}}
\label{app:specialfunctions}

As discussed in the main text, the FAB construction in \cref{sec:gaussianlikelihood} requires evaluating the marginal likelihood $f(y) = \int_\Theta f_\theta(y) \pi(d \theta)$ and, for the focal point, the posterior mean $\widehat{\theta}^{\FAB}(y) = \E[\theta \mid y]$.
Crucially, computing the FAB-CR relies on evaluating the weight function $w_{\alpha}(\theta_0)$ via a root-finding procedure over a grid of $\theta_0$ values.
This necessitates repeated evaluations of $f(y)$, making the computational cost of these functions an important consideration.
For the beta prime scale mixture priors in \cref{sec:vanishing_influence}, these quantities admit closed forms involving special functions of real arguments.

This section briefly discusses the special functions that arise \citep{Olver2010}, as well as strategies for their stable and efficient numerical evaluation.
For all simulations and experiments, we rely on the \href{https://docs.scipy.org/doc/scipy/reference/special.html}{\texttt{scipy.special}} implementation of the functions involved \citep{Virtanen2020} and perform computations on the log scale when possible.

\paragraph*{Kummer's confluent hypergeometric function $\OneFOne$ (\cref{prop:betaprime_marginal}).}

As shown in \cref{prop:betaprime_marginal}, the marginal likelihood $f(y)$ and the posterior mean $\widehat{\theta}^{\FAB}(y)$ for the beta prime scale mixture prior can be expressed in terms of Kummer's confluent hypergeometric function $\OneFOne$ as
\begin{equation}
    f(y) = \frac{1}{\sqrt{2 \pi \sigma^2}} \times \frac{\Gamma(a + 1/2) \Gamma(a + b)}{\Gamma(a) \Gamma(a + b + 1/2)} \times \OneFOne\left(a + \frac{1}{2}, a + b + \frac{1}{2}, -\frac{y^2}{2 \sigma^2} \right) \label{eq:betaprime_marginal}
\end{equation}
and
\begin{equation}
    \widehat{\theta}^\FAB(y) = y \times \left(1 - \frac{a + \frac{1}{2}}{a + b + \frac{1}{2}}\frac{\OneFOne\left(a + \frac{3}{2}, a + b + \frac{3}{2}, -\frac{y^2}{2 \sigma^2} \right)}{\OneFOne\left(a + \frac{1}{2}, a + b + \frac{1}{2}, -\frac{y^2}{2 \sigma^2} \right)}\right), \label{eq:betaprime_focal}
\end{equation}
respectively.
The function $\OneFOne(\alpha, \beta, z)$ is given by
\begin{equation*}
    \OneFOne(\alpha, \beta, z) = \sum_{k=0}^\infty \frac{(\alpha)_k}{(\beta)_k} \frac{z^k}{k!},
\end{equation*}
with $\alpha \in \mathbb{R}$, $\beta \in \mathbb{R} \setminus \{0, -1, -2, \dots\}$ and $z \in \mathbb{R}$, and where $(\alpha)_k$ denotes the rising factorial.
For the purpose of evaluating \cref{eq:betaprime_marginal,eq:betaprime_focal}, we are interested in the regime where $\alpha, \beta > 0$ and $z = -y^2 / (2 \sigma^2) \le 0$.
While $z \le 0$ avoids the exponential growth seen when $z \gg 0$, evaluating the naive power series for large negative $z$ is highly unstable due to the catastrophic cancellation of alternating signs.
In practice, we rely on \texttt{scipy.special.hyp1f1}, which is designed to evaluate ${}_1F_1$ stably across argument regimes, including large negative $z$.
However, this implementation relies on parameter-dependent algorithm switching, which can make $\OneFOne$ more computationally intensive than the simplified cases below when evaluated repeatedly within root-finding.
As discussed in \cref{sec:vanishing_influence}, several choices of $(a,b)$ lead to simplifications of \cref{eq:betaprime_marginal} in terms of more elementary functions, which help reduce the computational overhead of computing $f(y)$ repeatedly.

\paragraph*{Dawson's function (\cref{ex:horseshoe}).}
For $(a,b)= (1/2,1/2)$, the marginal likelihood $f(y)$ simplifies to
\begin{equation*}
    f(y) = \frac{2}{\pi^{3/2}} \frac{1}{|y|} D\left(\frac{|y|}{\sqrt{2 \sigma^2}}\right),
\end{equation*}
where $D(z)$ is Dawson's function, defined as
\begin{equation*}
    D(z) = e^{-z^2} \int_0^z e^{t^2} dt
\end{equation*}
for $z \in \mathbb{R}$.
The function $D(z)$ enters the expression for $f(y)$ proportionally to $D(t)/t$ with $t = |y| / \sqrt{2 \sigma^2}$.
Direct evaluation is numerically stable for most $t$, and we handle the removable singularity at the origin by directly substituting $\lim_{t \to 0} D(t)/t = 1$.

\paragraph*{Modified Bessel functions (\cref{ex:modified_bessel}).}
For $(a,b)=(1,1/2)$, the marginal likelihood $f(y)$ simplifies to
\begin{equation*}
    f(y) = \frac{1}{\sqrt{2 \pi \sigma^2}} \times \frac{\pi}{4} \times \exp\left(-\frac{y^2}{4 \sigma^2}\right) \times \left[I_0\left(\frac{y^2}{4 \sigma^2}\right) - I_1\left(\frac{y^2}{4 \sigma^2}\right)\right],
\end{equation*}
where $I_n(z)$ is the $n$-th order modified Bessel function of the first kind, given by
\begin{equation*}
    I_n(z) = \sum_{k=0}^\infty \frac{1}{k! \Gamma(k + n + 1)} \left(\frac{z}{2}\right)^{2k + n}
\end{equation*}
for $n \in \mathbb{N}$ and $z \in \mathbb{R}$.
The function $I_n(z)$ enters the expression for $f(y)$ through the term $e^{-t} (I_0(t) - I_1(t))$ with $t = y^2 / (4 \sigma^2)$.
Because $I_n(t)$ grows exponentially, direct evaluation will overflow for large $t$.
We avoid this by relying on the exponentially scaled version $e^{-t}I_n(t)$ provided by SciPy (\texttt{scipy.special.ive}) to ensure numerical stability.

\paragraph*{Elementary functions (\cref{ex:pareto}).}
For $(a,b) = (1/2, 1)$, the marginal likelihood $f(y)$ simplifies to an expression involving only elementary functions:
\begin{equation*}
    f(y) = \frac{\sigma}{\sqrt{2\pi}} \frac{1 - \exp(-y^2/(2\sigma^2))}{y^2}
\end{equation*}
for $y \neq 0$, and $f(0) = 1/(2\sigma\sqrt{2\pi})$.
While highly optimized libraries like SciPy make the empirical runtime difference between evaluating this and the special functions above small, this closed form circumvents the need for specialized routines entirely.
It represents an attractive, highly portable alternative for deploying robust FAB-CR procedures in software environments lacking comprehensive special function support.

\end{document}